\PassOptionsToPackage{dvipsnames}{xcolor}
\documentclass[11pt,letterpaper,twoside]{article} 
\usepackage{fullpage}
\usepackage{booktabs}
\usepackage{amsmath}
\usepackage{amsfonts}
\usepackage{amssymb}
\usepackage{amsthm}
\usepackage{upgreek}
\usepackage{times}

\usepackage{color}
\usepackage[final]{graphicx}
\usepackage{tikz}
\usepackage{algorithmic, algorithm}

\def\P{\mathbb{P}}

\def\R{\mathbb{R}}
\def\I{\mathbb{I}}
\def\eps{\varepsilon}
\def\1{\mathbf{1}}

\def\stocmode{0}

\def\arxivmode{0}
\def\fastmode{0}

\def\showauthornotes{0}

\def\showkeys{0}
\def\showdraftbox{1}
\def\showcolorlinks{0}
\def\usemicrotype{1}
\def\showfixme{1}

\ifnum\fastmode=0
\usepackage{etex}
\fi

\ifnum\fastmode=0
\usepackage[l2tabu, orthodox]{nag}
\fi

\ifnum\fastmode=0
\usepackage[T1]{fontenc}
\usepackage[full]{textcomp}
\fi

\usepackage[american]{babel}

\usepackage{mathtools}

\usepackage{amsthm}

\newtheorem{theorem}{Theorem}[section]
\newtheorem*{theorem*}{Theorem}

\newtheorem{proposition}[theorem]{Proposition}
\newtheorem*{proposition*}{Proposition}
\newtheorem{lemma}[theorem]{Lemma}
\newtheorem*{lemma*}{Lemma}
\newtheorem{corollary}[theorem]{Corollary}
\newtheorem*{conjecture*}{Conjecture}
\newtheorem{fact}[theorem]{Fact}
\newtheorem*{fact*}{Fact}

\newtheorem*{exercise*}{Exercise}

\newtheorem*{hypothesis*}{Hypothesis}

\theoremstyle{definition}
\newtheorem{definition}[theorem]{Definition}

\newtheorem{exercise-easy}[theorem]{Exercise}
\newtheorem{exercise-med}[theorem]{Exercise}
\newtheorem{exercise-hard}[theorem]{Exercise$^\star$}
\newtheorem{claim}[theorem]{Claim}
\newtheorem*{claim*}{Claim}

\newtheorem{remark}[theorem]{Remark}
\newtheorem*{remark*}{Remark}

\newtheorem*{observation*}{Observation}

\usepackage[
letterpaper,
top=0.7in,
bottom=0.9in,
left=1in,
right=1in]{geometry}

\ifnum\arxivmode=0
\usepackage{newpxtext} 
\usepackage{textcomp} 
\usepackage[varg,bigdelims]{newpxmath}
\usepackage[scr=rsfso]{mathalfa}
\usepackage{bm} 
\let\mathbb\varmathbb
\fi

\ifnum\arxivmode=1
\usepackage[varg]{pxfonts} 
\usepackage{textcomp} 
\usepackage[scr=rsfso]{mathalfa}
\usepackage{bm} 
\fi

\ifnum\showkeys=1
\usepackage[color]{showkeys}
\fi

\makeatletter
\@ifpackageloaded{hyperref}
{}
{\usepackage{hyperref}}

\definecolor{bleudefrance}{rgb}{0.01, 0.1, 1.0}
\definecolor{azure}{rgb}{0.0, 0.5, 1.0}

\ifnum\showcolorlinks=1
\hypersetup{
pagebackref=false,
colorlinks=true,
urlcolor=blue,
linkcolor=blue,
citecolor=OliveGreen}
\fi

\ifnum\showcolorlinks=0
\hypersetup{
pagebackref=false,
colorlinks=false,
pdfborder={0 0 0}}
\fi

\usepackage{prettyref}

\newcommand{\savehyperref}[2]{\texorpdfstring{\hyperref[#1]{#2}}{#2}}

\newrefformat{eq}{\savehyperref{#1}{\textup{(\ref*{#1})}}}
\newrefformat{lem}{\savehyperref{#1}{Lemma~\ref*{#1}}}
\newrefformat{def}{\savehyperref{#1}{Definition~\ref*{#1}}}
\newrefformat{thm}{\savehyperref{#1}{Theorem~\ref*{#1}}}
\newrefformat{cor}{\savehyperref{#1}{Corollary~\ref*{#1}}}
\newrefformat{cha}{\savehyperref{#1}{Chapter~\ref*{#1}}}
\newrefformat{sec}{\savehyperref{#1}{Section~\ref*{#1}}}
\newrefformat{app}{\savehyperref{#1}{Appendix~\ref*{#1}}}
\newrefformat{tab}{\savehyperref{#1}{Table~\ref*{#1}}}
\newrefformat{fig}{\savehyperref{#1}{Figure~\ref*{#1}}}
\newrefformat{hyp}{\savehyperref{#1}{Hypothesis~\ref*{#1}}}
\newrefformat{alg}{\savehyperref{#1}{Algorithm~\ref*{#1}}}
\newrefformat{rem}{\savehyperref{#1}{Remark~\ref*{#1}}}
\newrefformat{item}{\savehyperref{#1}{Item~\ref*{#1}}}
\newrefformat{step}{\savehyperref{#1}{step~\ref*{#1}}}
\newrefformat{conj}{\savehyperref{#1}{Conjecture~\ref*{#1}}}
\newrefformat{fact}{\savehyperref{#1}{Fact~\ref*{#1}}}
\newrefformat{prop}{\savehyperref{#1}{Proposition~\ref*{#1}}}
\newrefformat{prob}{\savehyperref{#1}{Problem~\ref*{#1}}}
\newrefformat{claim}{\savehyperref{#1}{Claim~\ref*{#1}}}
\newrefformat{relax}{\savehyperref{#1}{Relaxation~\ref*{#1}}}
\newrefformat{red}{\savehyperref{#1}{Reduction~\ref*{#1}}}
\newrefformat{part}{\savehyperref{#1}{Part~\ref*{#1}}}
\newrefformat{ex}{\savehyperref{#1}{Exercise~\ref*{#1}}}
\newrefformat{property}{\savehyperref{#1}{Property~\ref*{#1}}}

\newcommand{\Sref}[1]{\hyperref[#1]{\S\ref*{#1}}}

\usepackage{nicefrac}

\ifnum\fastmode=0
\ifnum\usemicrotype=1
\usepackage{microtype}
\fi
\fi

\ifnum\showauthornotes=2
\newcommand{\mynotes}[1]{{\sffamily\small\color{teal}{#1}}\medskip}
\newcommand{\Authornote}[2]{{\sffamily\small\color{blue}{[#1: #2]}}\medskip}
\newcommand{\Authornotecolored}[3]{{\sffamily\small\color{#1}{[#2: #3]}}}
\newcommand{\Authorcomment}[2]{{\sffamily\small\color{gray}{[#1: #2]}}}
\newcommand{\Authorstartcomment}[1]{\sffamily\small\color{gray}[#1: }

\newcommand{\Authorfnote}[2]{\footnote{\color{red}{#1: #2}}}
\newcommand{\Authorfixme}[1]{\Authornote{#1}{\textbf{??}}}
\newcommand{\Authormarginmark}[1]{\marginpar{\textcolor{red}{\fbox{\Large #1:!}}}}
\newcommand{\myexplain}[1]{{\sffamily\small\color{red}{\noindent [Explanation:\medskip\newline \begin{quote}#1\hfill]\end{quote}}}\medskip}

\else

\newcommand{\mynotes}[1]{}
\newcommand{\Authornote}[2]{}
\newcommand{\Authornotecolored}[3]{}
\newcommand{\Authorcomment}[2]{}
\newcommand{\Authorstartcomment}[1]{}

\newcommand{\Authorfnote}[2]{}
\newcommand{\Authorfixme}[1]{}
\newcommand{\Authormarginmark}[1]{}
\newcommand{\myexplain}[1]{}

\ifnum\showauthornotes=1
\renewcommand{\myexplain}[1]{{\sffamily\small\color{red}{\noindent \begin{quote}{\bf Explanation:} \medskip\newline #1\end{quote}}}\medskip}
\fi

\fi

\ifnum\showfixme=0

\fi

\usepackage{boxedminipage}

\newcommand{\Esymb}{\mathbb{E}}

\DeclareMathOperator*{\E}{\Esymb}

\newcommand{\textparen}[1]{\text{(#1)}}

\ifx\because\undefined
\newcommand{\because}[1]{\textparen{because #1}}
\else
\renewcommand{\because}[1]{\textparen{because #1}}
\fi

\newcommand\bdot\bullet

\ifx\mathds\undefined 

\else

\fi

\renewcommand{\leq}{\leqslant}

\renewcommand{\geq}{\geqslant}

\ifnum\showdraftbox=1

\else

\fi

\let\epsilon=\varepsilon

\numberwithin{equation}{section}

\newcommand\MYcurrentlabel{xxx}

\newcommand{\MYstore}[2]{%
  \global\expandafter \def \csname MYMEMORY #1 \endcsname{#2}%
}

\newcommand{\MYload}[1]{%
  \csname MYMEMORY #1 \endcsname%
}

\newcommand{\MYnewlabel}[1]{%
  \renewcommand\MYcurrentlabel{#1}%
  \MYoldlabel{#1}%
}

\newcommand{\MYdummylabel}[1]{}

\newcommand{\torestate}[1]{%
  \let\MYoldlabel\label%
  \let\label\MYnewlabel%
  #1%
  \MYstore{\MYcurrentlabel}{#1}%
  \let\label\MYoldlabel%
}

\newcommand{\restatetheorem}[1]{%
  \let\MYoldlabel\label
  \let\label\MYdummylabel
  \begin{theorem*}[Restatement of \prettyref{#1}]
    \MYload{#1}
  \end{theorem*}
  \let\label\MYoldlabel
}

\newcommand{\restatelemma}[1]{%
  \let\MYoldlabel\label
  \let\label\MYdummylabel
  \begin{lemma*}[Restatement of \prettyref{#1}]
    \MYload{#1}
  \end{lemma*}
  \let\label\MYoldlabel
}

\newcommand{\restateprop}[1]{%
  \let\MYoldlabel\label
  \let\label\MYdummylabel
  \begin{proposition*}[Restatement of \prettyref{#1}]
    \MYload{#1}
  \end{proposition*}
  \let\label\MYoldlabel
}

\newcommand{\restatefact}[1]{%
  \let\MYoldlabel\label
  \let\label\MYdummylabel
  \begin{fact*}[Restatement of \prettyref{#1}]
    \MYload{#1}
  \end{fact*}
  \let\label\MYoldlabel
}

\newcommand{\restate}[1]{%
  \let\MYoldlabel\label
  \let\label\MYdummylabel
  \MYload{#1}
  \let\label\MYoldlabel
}

\newcommand{\addreferencesection}{
  \phantomsection
\ifnum\stocmode=0
  \addcontentsline{toc}{section}{References}
\else
  \addcontentsline{toc}{section}{References \hspace*{1in} --------- End of extended abstract ---------}
\fi

}

\let\origparagraph\paragraph
\renewcommand{\paragraph}[1]{\vspace*{-10pt}\origparagraph{#1.}}

\allowdisplaybreaks

\usepackage{paralist}

\usepackage{comment}

\usepackage{enumitem}

\begin{document}
\title{Hashing-Based-Estimators for  Kernel Density \\in High Dimensions\footnote{A preliminary version of this paper appeared in FOCS'2017~\cite{charikar2017hashing}.}\\
}
\author{Moses Charikar\\Stanford University \\Department of Computer Science\\\texttt{moses@cs.stanford.edu} \and Paris Siminelakis\\
Stanford University\\Department of Electrical Engineering\\\texttt{psimin@stanford.edu}}
\maketitle
\thispagestyle{empty}
\begin{abstract}
Given a set of points $P\subset \R^{d}$ and a kernel $k$, the Kernel Density Estimate at a point $x\in\R^{d}$  is defined as $\mathrm{KDE}_{P}(x)=\frac{1}{|P|}\sum_{y\in P} k(x,y)$. We study the problem of designing a data structure that given a data set $P$ and a kernel function, returns \emph{approximations to the kernel density} of a query point in \emph{sublinear time}.  We introduce a class of unbiased estimators for kernel density  implemented through locality-sensitive hashing, and give general theorems bounding the variance of such estimators.
These estimators give rise to   efficient data structures for estimating the kernel density in high dimensions for a variety of commonly used kernels. Our work is the first to provide data-structures with theoretical guarantees that improve upon simple random sampling in high dimensions.
\end{abstract}
\newpage
\tableofcontents
\newpage
\setcounter{page}{1}

\section{Introduction}
A fundamental question in Statistics and Learning Theory is the following: \emph{
given a set of points $P\subset \R^{d}$ sampled from some unknown distribution $\mathcal{D}$ estimate the probability at an arbitrary point $x\in \R^{d}$.}
This problem is known as \emph{density estimation} and different ways to formalize it lead to very different statistical and computational tasks. In the past two decades the problem has attracted significant interest in theoretical computer science~\cite{chan2014efficient,acharya2017sample}. Some of the most important problems that have been studied are learning discrete distributions~\cite{kearns1994learnability}, learning mixture models~\cite{vempala2004spectral},  and more recently the topic of robust estimation in high dimensions~\cite{diakonikolas2016robust,lai2016agnostic}. 
In this paper, we focus on \emph{Kernel Density Estimation} (KDE), one of the most widely developed methods in non-parametric estimation.

\subsection{Kernel Density Estimation}

 In this approach, given a set of $n$  points $P$, starting from the empirical distribution $\tilde{\mu}(x)=\frac{1}{n}\sum_{y\in P}\delta_{y}$, one obtains a smooth distribution by ``convolving'' it with a kernel function $k$, whose smoothness is typically controlled by a parameter $\sigma>0$ called the \emph{bandwidth}. 
\begin{definition}[Kernel Density] Given a kernel function $k_{\sigma}:\R^{d}\times \R^{d}\to [0,1]$  and a dataset $P\subset \R^{d}$ we define the Kernel Density (KD) of $P$ at a point $x\in \R^{d}$ as:
\begin{equation}
\mathrm{KDE}_{P}(x):= \frac{1}{|P|} \sum_{y\in P} k_{\sigma}(x, y) 
\end{equation}
\end{definition}
This is a natural way to extend the function smoothly from a discrete set of points to the whole space that is independent of any particular parametric assumption on the underlying distribution of the data. Selecting the kernel and bandwidth are intensively studied subjects in the literature of non-parametric estimation~\cite{devroye2012combinatorial} for which there is still ongoing theoretical research~\cite{goldenshluger2011bandwidth}. The kernel function $k(x,y)$ is typically a function of only $x-y$ (shift invariant kernels) or just the euclidean distance $\|x-y\|$ (radial kernels). One of the most prominent functions is the \textbf{Gaussian kernel}
\begin{equation}
k_{\sigma}(x,y) = \exp\left(-\frac{\|x-y\|^{2}}{\sigma^{2}}\right)
\end{equation}

The importance of KDE lies in that it gives a simple and general way of approximating the underlying probability distribution that can be subsequently used to perform more complex and computationally intensive tasks.
Examples include {\em mode estimation}~\cite{arias2015estimation},
{\em outlier detection}~\cite{schubert2014generalized},
{\em local regression}~\cite{fan1996local},
{\em reproducing kernel Hilbert spaces}~\cite{scholkopf2001learning},
{\em density based clustering}~\cite{rinaldo2010generalized}, and
{\em topological data analysis}~\cite{joshi2011comparing,chen2016density}.
Kernel Density Estimation is consequently an important primitive that is a building block in many applications.

In all of the above settings, at some point, the following problem is solved:\emph{
given  $P\subset \R^{d}$,  $z\in \R^{n}$, compute  $\mathrm{KDE}^{z}_{P}(x) := \sum_{i=1}^{n}k(x,y_{i}) z_{i}$.} 
This can be computed exactly in linear time, but this is prohibitively slow for large data sets, especially since it is needed repeatedly in applications of interest.

The problem has been studied extensively in the {\em batch} setting, where given a set of $n$ points, the goal is to compute, for each of the points, a sum of contributions due to all the points, i.e. $n$ queries of the above form. Such computations are prevalent in the field of scientific computing and involve computing approximations to $y = Kz$ where $K$ is an $n\times n$ kernel matrix.
In low dimensions, Fast Multipole Methods (combining  hierarchical space partitions with Taylor approximations) were developed~\cite{lee2006dual} to reduce the trivial $O(n^2)$ runtime to $O(n \log n)$ (and $O(n)$ in some cases).
The fast multipole algorithm~\cite{greengard1987fast} has been enormously influential in numerical analysis and scientific computing; it was named as one of the top 10 algorithms of the 20th century by the editors of {\em Computing in Science and Engineering}~\cite{dongarra2000guest}. For this work, Greengard and Rokhlin received the 2001 Steele Prize. The KDE problem thus lies in the core of both scientific computing as well as machine learning.

In this paper, we study the problem of \emph{approximately} computing the KDE.  For most of the paper we fix $k_{\sigma}$ to be the Gaussian Kernel and define the following computational problem:

\begin{definition}[KDE Problem] Given a dataset $P\subset \R^{d}$ of $n$ points, and $\epsilon,\delta,\tau \in [0,1]$ construct a  data structure that given a query $x\in \R^{d}$ with $\mathrm{KDE}_{P}(x) = \mu \in [\tau,1]$ returns a number $\hat{\mu}$ such that $\P[|\hat{\mu}-\mu|\geq \epsilon \cdot  \mu]\leq \delta$. We call this problem the \emph{$(\mu,\epsilon,\delta)$-KDE problem}.
\end{definition}

\subsection{Our Contribution}

The starting point of our work is the old and tested idea of \emph{importance sampling}. Given non-negative weights $w_{1},\ldots, w_{n}$ and a distribution $Q$ over $[n]$ (inducing probabilities $q_{1},\ldots,q_{n}$), an \emph{unbiased estimator} for $\mu:=\frac{1}{n}\sum_{i=1}^{n}w_{i}$ is given by sampling $I\in [n]$ according to $Q$ and returning $Z = w_{I}/(q_{I}n)$. 
The minimum variance estimator is obtained by setting $q_{i}^{*}= w_{i}/\sum_{j}w_{j}$ for which the variance is zero. What precludes us from obtaining such probabilities is that in our setting of KDE, the weights $w_{i}=w_{i}(x):= k(x,x_{i})$ depend  on the query $x$, and thus the sampling distribution $Q$ needs to be adaptive to the query. Furthermore, having  an ideal distribution $Q^{*}$ indirectly involves knowing the \emph{normalizing constant} $\sum_{i}w_{i}=n\mu$, the very quantity we wish to estimate. Thus, the main challenge in turning the idea of importance sampling into an algorithm is to have an \emph{efficient way} to define an \emph{adaptive sampling distribution $Q(x)$} that has \emph{low variance}. We next present our methods at an abstract level and subsequently show their implications for Kernel Density Estimation.

\subsubsection{Hashing-Based-Estimators (HBE)}
Our main contribution is to introduce a hashing-based framework to \emph{succinctly} define  an adaptive distribution $Q(x)$ and to   provide sharp tools  \emph{bounding the variance} of the resulting unbiased estimator. Our estimators are formed \nolinebreak  by:

\begin{itemize}
\item \textbf{Preprocessing:} given a hash function $h$ sampled from a hash family $\mathcal{H}$, where the colision probability of 
$x,y$ is $p(x,y)$, we evaluate $h$ on   $x_1,\ldots,x_n$ and form the corresponding hash table $H$.
\item \textbf{Querying:} let  $H(x)\subseteq P$ denote  the cell where the query $x$ falls into and let $y\in P$ be a random element of $H(x)$, we return $Z_{h}(x) =\frac{1}{n} \frac{k(x,y)}{p(x,y)}|H(x)|$ (or $0$ if $H(x)$ is empty).
\end{itemize}
We say that a HBE has \emph{complexity} $T$ if the evaluation time (resp. space usage/preprocessing)  is bounded by $T$ (resp. $T\cdot n$). 
This \emph{two-level sampling} procedure induces  probabilities that  depend both on $P$ and $\mathcal{H}$. Although $|H(x)|$ is \emph{a priori} a random variable, it becomes known to us through the preprocessing step. This  sidesteps the issue of computing the normalizing constant   separately for each query and is at the core of our approach. By sampling a number of hash functions $h_{1},\ldots, h_{m}$ and creating hash tables $H_{1},\ldots, H_{m}$, at query time we can  produce $m$ independent samples that are used by the median-of-means procedure to produce an accurate estimate of the mean $\mu$.

The challenge in fully implementing this scheme is to bound the second moment of our estimator. In this regard, picking a random element from $H(x)$ turns out to be crucial, as it is this step that allows us to get an analytical handle on the second moment of $Z_{h}$.  We provide two general theorems that bound the variance of \emph{Hashing-Based-Estimators}. The first theorem applies to any  HBE.
\begin{theorem}[Two-points suffice]\label{thm:two-points}
 Up to absolute constants the variance of a HBE is maximized by  datasets where there are \emph{only two values} for the weights and sampling probabilities. 
\end{theorem}
This characterization is extremely useful as  
given $p_{i}(x):=p(x,x_{i})$ and $w_{i}(x)$, 
it reduces bounding the variance to a simple case analysis. The bound depends on the compatibility between the probability and weights, and is captured by the maximum element of an explicitly defined matrix. Going beyond the general case we identify a natural class of HBE that induce sampling probabilities that vary as a power of the weights.
\begin{definition}\label{def:scale-free}
An HBE is \emph{$(\beta,M)$-scale free}
for parameters $\beta\in(0,1]$ and $M\geq 1$ if, for all $i\in [n]$ and $x$, 
$M^{-1}\cdot w_{i}^{\beta}(x) \leq p_{i}(x) \leq M\cdot w_{i}^{\beta}(x)$.
\end{definition}
The second theorem  provides a refined analysis of the variance of scale-free HBE that is able to capture additional structure when one exists. In particular, the upper bound on the variance improves when most of the contribution to $\mu=\mu(x)$ comes from relatively large weights $w_{i}$. We state here the weakest bound that assumes nothing about the weights. 
\begin{theorem}\label{thm:scale-free-var}
 Let $Z$ be an unbiased $(\beta,M)$-scale free HBE for $\beta \in[\frac{1}{2},1]$. Then $\E[Z^{2}] \leq \mu^{2} \cdot (4M^{3}/\mu^{\beta})$.
\end{theorem}

The main technical tools behind the analysis are two H\"{o}lder-type inequalities that we develop. The first one (Lemma \ref{lem:two-sided}) is a simple two-sided (matrix) extension of H\"{o}lder's inequality that bounds a quadratic form over the intersection of two weighted $\ell_{1}$-balls. A clever application of this inequality gives us the proof of  Theorem \ref{thm:two-points}.  The second  (Lemma \ref{lem:mon-holder}) is a non-trivial H\"{o}lder-type inequality for monotone vectors that in combination with some easy consequences of H\"{o}lder's inequality gives us the refined analysis of \emph{scale-free} estimators. The proof also shows the possibility of exploiting other structural assumptions of the data. This is important in the statistical setting (where the data set is sampled from a distribution) or for parameter tuning in practice.

Our theorem shows that under no structural assumptions, the optimal choice is $\beta^{*}=1/2$ and results in an estimator with relative variance $ \frac{\mathrm{Var}[Z]}{(\E[Z])^{2}}\leq\mathrm{V}(\mu) = 4M^{3}\cdot  \mu^{-1/2}$. Using the \emph{Median-of-means} (MoM) technique,
we can estimate $\mu$ within a multiplicative accuracy of $(1\pm \eps)$ using $O\left(M^{3}\frac{1}{\epsilon^{2}}\frac{1}{\sqrt{\mu}}\log(1/\delta)\right)$ independent samples. Observe that this does not directly imply an algorithm to estimate $\mu$ using this many samples, as setting the number of samples (sufficient for accurate estimation) requires approximate knowledge of $\mu$, the very quantity we are aiming to estimate. 

We resolve this issue by proposing an adaptive procedure that uses $O(1)$ times additional samples to get a constant factor approximation to $\mu$. 
We start with an overestimate of $\mu$ and iteratively decrease it until we get close enough to the truth where a consistency check is satisfied. The resulting algorithm, \emph{Adaptive Mean Relaxation}, is applicable to settings where one has an unbiased estimator whose (upper bound on) variance is a non-decreasing function of the mean $\mu$ and the relative variance is a decreasing function of the mean. 
 
\subsubsection{Kernel Density Estimation through Locality Sensitive Hashing}

We next turn to address the $(\mu,\epsilon,\delta)$-KDE problem. To provide intuition about the problem we first analyze two simple randomized estimators, the (uniform) Random Sampling (RS) estimator and an estimator based on Random Fourier Features (RFF)~\cite{rahimi2007random}. We show that in the worst case, the first has variance bounded by $\mu$ whereas the RFF estimator has constant variance. 
Using the MoM framework and our adaptive procedure, one immediately gets an algorithm to solve the KDE problem using $O(\min\{\frac{1}{\epsilon^{2}}\frac{1}{\mu}\log(1/\delta),n\})$ samples that is polynomial in $(\epsilon,\mu,\log(1/\delta))$. 

In order to improve upon this simple bound, we employ the framework of Hashing-Based-Estimators instantiated with \emph{Locality Sensitive Hashing} schemes. All our results follow a similar theme: (a) we obtain pointwise upper and lower bounds on the collision probabilities of the hash functions, and (b) we bound the variance of the estimator by invoking either Theorem \ref{thm:two-points} or Theorem \ref{thm:scale-free-var} in cases where we are able to obtain scale-free estimators (cf. Table \ref{my-label}). 

\begin{theorem}[Informal] \label{thm:meta}
There exist scale-free HBE for the Gaussian, Exponential and $t$-Student kernels.
\end{theorem}
\begin{table}[]
\centering
\label{my-label}
\caption{Scale free estimators for KDE using LSH }
\begin{tabular}{@{}lcc@{}}
\toprule
 Kernel                      & \multicolumn{1}{c}{$(\beta,M)$} & \multicolumn{1}{c}{Complexity $T$} \\ \midrule
 $e^{-\|x-y\|^{2}}$    &   $(\beta,e^{O(R^{\frac{4}{3}}\log\log n)})$                       & $ e^{O(R^{\frac{4}{3}}\log\log n)}$                              \\
$e^{-\|\mathbf{x}-y\|}$           & $(\beta,\sqrt{e}\ )$                      & $O(dR^{2})$                                      \\
$\frac{1}{1+\|\mathbf{x}-y\|_{2}^{p}}$ & $(\frac{q}{p},3^{q})$                             & $O(dp)$                                       \\ \bottomrule
\end{tabular}

\end{table}

Using our  theorem for scale-free estimators and the adaptive algorithm, we arrive at our main result for the KDE problem.

\begin{theorem}\label{thm:main-kernel}
For a kernel $k$ and dataset $P$ for which there exists a $(\frac{1}{2},M)$-scale-free estimator with complexity $T$,  there exists a data structure that solves the $(\mu,\epsilon,\delta)$-KDE problem $\forall \mu\in[\tau,1]$ using $O(M^{3} \frac{1}{\epsilon^{2}}\frac{1}{\sqrt{\mu}}\log(1/\delta)T)$ time and  $O(M^{3} \frac{1}{\epsilon^{2}}\frac{1}{\sqrt{\tau}}\log(1/\delta)T\cdot n)$ space.
\end{theorem}

As an application, we show that one can use such a data structure to get an \emph{approximate vector-matrix multiplication algorithm} for Kernel matrices using time that is adaptive to the vector and is always bounded by $\frac{n^{1+o(1)}}{\sqrt{\tau^{'}}}\frac{1}{\epsilon^{2}}$ where $\epsilon,\tau^{'}$ indicate respectively the relative and additive error per coordinate (cf. Theorem \ref{thm:matrix}). 
This result is important as it improves on the main bottleneck of Kernel Ridge Regression~\cite{avron2016faster}, that is, multiplying a dense Kernel matrix with a vector and requires time $O(n^{2})$ in general.

\subsubsection{Lower Bounds for KDE problem}

We also complement our results by providing a reduction between hard instances for the Approximate Nearest Neighbor Search problem to the $(\mu,\epsilon,\delta)$-KDE problem. We show that the latter is at least as hard as the $(r,c)$-ANNS with $n=\frac{1}{\mu}$ points and $c=O(\frac{\log(n)}{\log(1/\epsilon)})$ under the Hamming distance. Combined with the results of Panigrahy, Talwar, Wieder~\cite{panigrahy2010lower} and Andoni et al~\cite{andoni2017optimal}, we get non-trivial lower bounds in the \emph{cell-probe model}  with a \emph{single probe} that captures an interesting class of algorithms based on \emph{adaptive coresets}. 

\subsection{Related Work}

The problem of Kernel Density Estimation although widely studied in low-dimensions~\cite{greengard1991fast} has largely been unexplored in high dimensions~\cite{march2015askit}.
In recent parallel and independent work, Spring and Shrivastava~\cite{spring2017new} introduced the idea of using locality sensitive hashing as a sampling scheme 
to estimate the \emph{partition function} of log-linear models, 
albeit without theoretical guarantees.

\paragraph{Coresets} The problem of KDE has mostly been investigated in the context of \emph{coresets}. The first theoretical work we are aware of is that of Phillips~\cite{phillips2013varepsilon}. Given a kernel $k$ and a set $P\subset \R^{d}$, a subset $S\subset P$ is an  $\epsilon$-coreset if $\left|\mathrm{KDE}_{P}(x) - \mathrm{KDE}_{S}(x) \right| \leq \epsilon$ for all $x\in \R^{d}$.  
Phillips uses techniques from discrepancy theory to show that one can construct an $\epsilon$-coreset of size $O\left(1/\epsilon^{2-\frac{4}{d+2}}\log^{1-\frac{2}{d+2}}(1/\epsilon\delta)\right)$ with probability at least $1-\delta$. For large $d$ the bound deteriorates and becomes similar to what one gets by simple random sampling $O(\frac{1}{\epsilon^{2}}\log(1/\delta))$ which is known to be tight. For relative error, recent work~\cite{Zheng} uses random sampling to give a similar guarantee that roughly requires $O(\frac{1}{\epsilon^{2}}\frac{1}{\mu})$ samples. Our lower bound against the $(m,w,1)$-cell probe models encompasses algorithms of this kind, and shows that any set $S$ (not necessarily a subset of $P$) that can be used to answer the $(\mu,\epsilon,\delta)$-problem must have size at least $\Omega(\frac{1}{\mu})$. The implication is that in terms of constructing a  \emph{coreset} for KDE in high dimensions, Random Sampling is essentially optimal.

\paragraph{Kernel Matrix-Vector Multiplication} A closely related idea to that of coresets, is that of Nystr\"{o}m approximation. In this method, given a kernel matrix $K $ a set of $s$ columns (points) is selected and subsequently $K$ is projected on their span. 
Musco and Musco~\cite{musco2017recursive} propose a method based on recursive leverage-score sampling that, using $s=\Theta(k\log k)$ points, $\tilde{O}(nk^{2})$ time and $\tilde{O}(nk)$ space,   outputs a matrix $\tilde{K}$ such that $\|\tilde{K}-K\|_{2}\leq \lambda$ with $\lambda = \frac{1}{k}\sum_{i=k+1}^{n} \sigma_{i}(K)$ . One can use this algorithm to obtain an approximate Kenrel-Matrix Vector Multiplication algorithm $\hat{y}\approx Kz$.  For kernel matrices of large rank, like the ones corresponding to the equilateral metric of $r=\sqrt{n}=\tau^{-1}$ clusters consisting of $\sqrt{n}$ identical points, their algorithm requires $\tilde{O}(n^{2})$ time and space $\tilde{O}(n^{3/2})$ whereas our algorithm requires $\tilde{O}(n^{\frac{5}{4}})$ time and space  for the Exponential and Polynomial kernel, and $n^{\frac{5}{4}+o(1)}$ for the Gaussian kernel. However,  the algorithm of Musco and Musco applies to any kernel and gives guarantees that hold simultaneously for any test vector $z$, whereas our method applies for a single (non-negative) vector $z$.

\subsection{Open questions}

Our work leaves open a few intriguing directions. 
\paragraph{Data-dependent Hashing} There is a recent line of work~\cite{andoni2015optimal,andoni2017optimal} that designs data-dependent LSH schemes that are  optimal within a certain class of hashing-based algorithms. We believe that modifications of such schemes can be used for the purposes of KDE.

\paragraph{Batch Setting}  The study of the offline or batch setting for Nearest Neighbor Search has received renewed interest over the past  years and has brought a wealth of techniques into  light~\cite{valiant2012finding,alman2015probabilistic}. Given the connection between KDE and the Nearest Neighbor Search problem, it would be of practical and theoretical interest to design data-structures that offer provable speedups for the offline setting of the KDE problem. In the regime where $\epsilon=\exp(-\omega(\log^{2}(n)))$ there is recent work~\cite{backurs2017fine} that shows that under SETH no significant improvements can be made beyond quadratic time.

\paragraph{Applications of HBE} It would be interesting to explore other extensions of HBE, e.g. get theoretical guarantees for estimating the \emph{partition function} of log-linear models~\cite{spring2017new} or analyze the performance of multi-probe schemes.

\subsection{Outline of the paper}
In Section \ref{sec:prelims}, we present some preliminaries and describe the general approach of deriving unbiased estimators and bounding their variance through H\"{o}lder-type inequalities, whose proofs are presented in Section \ref{sec:holder}. In Section \ref{sec:hbe}, we introduce Hashing-Based-Estimators and give two general tools to bound their variance. In Section \ref{sec:amr}, we present a framework for adaptively estimating the mean of a random variable when we have a known upper bound on the variance that depends on the unknown mean. In Sections \ref{sec:euclidean} and \ref{sec:gaussian}, we use Locality Sensitive Hashing schemes to derive unbiased estimators for the Exponential, $t$-Student and Gaussian Kernels that when instantiated within the framework of Section \ref{sec:amr} give rise to efficient data structures for the KDE problem. Using these data structures we then show in Section \ref{sec:matrix} how to perform approximate Fast Kernel-Matrix Vector multiplication. In Section \ref{sec:lower}, we present a lower bound for KDE for the Gaussian kernel in the cell-probe model.

\section{Preliminaries}\label{sec:prelims}
\paragraph{Notation} For a vector $x\in \R^{n}$ let $\|x\|^{p}_{p}:=\sum_{i=1}^{n}|x_{i}|^{p}$ for $p\geq 1$ denote the $\ell_{p}$-norm. For a strictly positive vector $w\in \R_{++}^{n}$, we denote by $\|x\|_{w,1}:= \sum_{i}w_{i}|x_{i}|$ the weighted $\ell_{1}$-norm. Given a number $\tau>0$ let $(x_{i})_{\tau}:=\max\{x_{i},\tau\}$. 
 Given a probability distribution $\nu$, we write $Y\sim \nu$ to denote that $Y$ is sampled from $\nu$. For a set $S$, $U(S)$ denotes the uniform distribution over $S$ and  $S^{\otimes k}$ denotes the Cartesian product of $S$ with itself $k\geq 1$ times. Similarly, $\nu^{\otimes k}$ denotes the product distribution $\nu\times \ldots \times \nu$. We use $N(0,I_{d})$ to denote the standard multivariate normal distribution and $\Phi(\cdot)$ its CCDF. Throughout the  paper $P=\{x_{1},\ldots,x_{n}\}\subset \R^{d}$ will denote  a set of $n$ points from $\R^{d}$ and we shall assume that $\mathrm{diam}(P\cup \{x\}) \leq R$ for any query $x\in \R^{d}$ where $R=R(n)>0$.
\subsection{$V$-bounded Estimators and Median-of-Means}
 For $\epsilon,\delta>0$, 
given a query point $x\in \R^{d}$ and weights $w_{1}(x),\ldots,w_{n}(x)$ induced by a set $P$, our goal is to estimate $\mu(x):=\frac{1}{n}\sum_{i=1}^{n}w_{i}(x)$ within a multiplicative $(1\pm \epsilon)$ accuracy with probability at least $1-\delta$. When it is clear from the context we will often drop the dependence on $x$ and simply write $w_{i}$ and $\mu$. An estimator $Z\sim \nu$ is called \emph{unbiased} for our problem if $\E[Z] = \mu$. The basis of our approach is to design an unbiased estimator that has small variance relative to $\mu$.

\begin{definition} Given a non-increasing function $V:\R\to \R_{+}$, we call an unbiased estimator $Z$,  $V$-bounded if $\E[Z^{2}] \leq \mu^{2}\cdot V(\mu)$ and $\mu^{2}V(\mu)$ is non-decreasing.
\end{definition}

The function $V$ is intimately related to the number of samples need to estimate $\mu$ and is often referred to as (a bound on) the \emph{relative variance}. The principal approach to use such an estimator is the Median-of-Means (MoM) technique~\cite{alon1996space}, that allows one to get an estimate $Z_{\epsilon,\delta}$ such that $ \P\left[\left|Z_{\epsilon,\delta} - \mu \right| \geq \epsilon \cdot\mu \right] \leq \delta$ using $O(\frac{1}{\epsilon^{2}}V(\mu)\log(\frac{1}{\delta}))$ samples.

\begin{algorithm}[h]
	\caption{Median-of-Means (MoM)}
	\label{alg:example}
	\begin{algorithmic}[1]
		\STATE {\bfseries Input:}   Estimator $Z\sim \nu$, $V\geq 0$, accuracy $\epsilon\in (0,1)$, success prob. $\delta\in(0,1)$.
		\STATE $m(\epsilon,V) \gets \lceil \frac{6}{\epsilon^{2}}V\rceil$, $L(\delta) \gets \lceil 9\log(1/\delta)\rceil$.
		\STATE $Z^{(i)}_{j} \stackrel{iid}{\sim} \nu$ for $j=1,\ldots, m$, and $i=1,\ldots, L$.
		\STATE $Z^{(i)} \gets \mathrm{mean}\{Z_{1}^{(i)},\ldots,Z_{m}^{(i)}\}$ for $i=1,\ldots, L$.\label{code:means}
		\STATE {\bf Output:} $Z_{\epsilon,\delta} \gets \mathrm{median}\{Z^{(1)},\ldots, Z^{(L)}\}$
	\end{algorithmic}
\end{algorithm}

Most of our effort in this paper goes into obtaining efficient $V$-bounded estimators with $V$ being of the form $V(\mu)= n^{o(1)} \cdot\mu^{-\Delta}$
for some $0\leq \Delta\leq 2$.   Once we have such an estimator, we will be able to combine it with the MoM technique and eventually get an estimation algorithm. Hence, the main challenge is bounding the variance.

\subsection{Bounding the Variance}

Given an unbiased estimator $Z$ in order to bound the variance we first obtain a simple data-dependent upper bound on $\E[Z^{2}] \leq F(w,P)$ for some function $F$, and then for a class of datasets $\mathcal{P}$, we aim to show that $\sup_{P\in\mathcal{P}} \{F(w,P)|\sum_{i}w_{i}=n\mu\} \leq\mu^{2} V(\mu)$. 

\subsubsection{Baseline estimators}
As a warm up, we present two simple unbiased estimators for the Kernel Density Problem and bound their variance.  Let $Y\sim U[P]$,  the \emph{random sampling} estimator is given by $Z_{\mathrm{RS}}(x) = k(x,Y)$.  Let $\theta\sim U[0,\pi]$ and $g\sim N(0, I_{d})$, the
\emph{Random Fourier Features}~\cite{rahimi2007random} estimator is given by $
Z_{RFF}(x) = \frac{2}{|P|}\sum_{y\in P} \left(\cos(g^{\top}x+\theta) \cos(g^{\top}y+\theta) \right)
$.
\begin{proposition} The RS and RFF estimators are unbiased and satisfy respectively
$\E[Z^{2}_{RS}] \leq \mu^{2}\cdot \mu^{-1}$ and $
\E[Z^{2}_{RFF}]  \leq \mu^{2}\cdot 4\mu^{-2}$. 
Moreover, the bounds  are tight up to constants in the worst case.
\end{proposition}
\begin{proof}[Proof Sketch] The fact that the RS estimator is unbiased is trivial, whereas the fact that RFF is unbiased was shown by Rahimi and Recht~\cite{rahimi2007random} and follows from Bochner's theorem  and trigonometric identities. We next bound the second moment
\begin{align*}
\E[Z_{RS}^{2}] &\leq \max_{y\in P}\{k(x,y)\}\cdot \frac{1}{|P|}\sum_{y\in P} K(x,y) \leq  \mu^{2}\cdot \mu^{-1}\\
\E[Z_{RFF}^{2}]  &\leq \left(\frac{2}{|P|}\sum_{y\in P} 1\right)^{2} = 4 \leq \mu^{2} \cdot4 \mu^{-2}
\end{align*}
To see that these bounds are tight up to constants consider for the RS estimator a dataset with  $n\mu$ points located at $x$ and the rest $n(1-\mu)$ points at distance $\sqrt{\log(1/\mu)}$ from $x$. For the RFF estimator the worst case is when all points are all located at the same point $y_{0}$ at distance $\sqrt{\log(1/\mu)}$. 
\end{proof}
We expand more on the RFF estimator in Appendix \ref{sec:RFF}.
The result on Random Sampling shows that KDE problem is solvable in time $O\left(\min\left\{\frac{1}{\epsilon^{2}}\frac{1}{\mu}\log(\frac{1}{\delta}),n\right\}\right)$. 

\subsubsection{Upper bounds via H\"{o}lder-type inequalities}
Despite the simplicity of the above result, some salient features of the problem are revealed. Firstly, the quality of the initial upper bound $F(w,P)$ can differ dramatically between two different estimators. Secondly, despite the simplicitly of the analysis, often the resulting bounds are tight up to constants. Lastly, in both cases we used \emph{H\"{o}lder's inequality} to get the bound. This is going to be a general theme as behind all our bounds on the variance are increasingly sophisticated consequences of H\"{o}lder's inequality. We state below the two main inequalities used to bound the variance of \emph{Hashing-Based-Estimators}.

\begin{lemma}[two-sided Holder]\label{lem:two-sided}
 Let  $v,w\in \R^{n}$ be strictly positive vectors, then for any two sets $S,S^{'}\subseteq [n]$:
\begin{equation}
\sum_{i\in S, j\in S^{'}}A_{ij}x_{i}x_{j} \leq  
   \|x\|_{v,1}\|x\|_{w,1} \cdot \max_{i\in S,j\in S^{'}}\left\{\frac{|A_{ij}|}{v_{i}w_{j}}\right\}
\end{equation}
\end{lemma}

\begin{lemma}[Monotone H\"{o}lder]\label{lem:mon-holder} $ \forall n\geq 1, \ \beta\in[\frac{1}{2},1]$,    and   $\forall x\in \R^{n}$ such that   $|x_{1}|\geq |x_{2}|\geq \ldots \geq |x_{n}|$  , we have $
 \sum |x_{i}|^{\frac{2-\beta}{\beta}}\left(i+\sum_{j>i}\frac{|x_{j}|}{|x_{i}|}\right) \leq n^{\beta}\cdot  \left(\sum_{i=1}^{n}|x_{i}|^{\frac{1}{\beta}}\right)^{2-\beta}$
with equality holding for $x^{*} = c\1$ for any $c\neq 0$.
\end{lemma}

\section{Importance Sampling through Hashing Based Estimators}\label{sec:hbe}
Given a distribution $\nu$ over a collection of hash functions $\mathcal{H}$, let $h\in \mathcal{H}$ be an element sampled according to $\nu$. For a point $x\in \R^{d}$ let $H(x) := \{y\in P: h(y)=h(x) \}$ be the set of elements in $P$ that have the same hash value as $x$. Also, let $I(x)$ be a uniform random element out of $H(x)$ or $\perp$ if $H(x)$ is empty.  Further, define $p_{i} := \P[i\in H(x)]$ and $\tilde{P} :=\{i\in P| p_{i}>0\}$. We also set $p_{\perp}:=1$ and $w_{\perp}:=0$.  An $(\mathcal{H},\nu)$-hashing based estimator (HBE) is defined as $Z_{h}=Z_{h}(x) := \frac{w_{I(x)}}{p_{I(x)}}\cdot  \frac{|H(x)|}{n}$.
\subsection{Moments of HBE}
\begin{lemma}[Moments]\label{lem:var-basic} Let $Z_{h}$ be a $(\mathcal{H},\nu)$-HBE and $P$ a set of points. Let $p_{1}\geq p_{2}\geq \ldots \geq p_{n}$, then 
\begin{align}
\E[Z_{h}]  = \frac{1}{n} \sum_{i\in \tilde{P}} w_{i}\qquad \text{and} \qquad 
\E[Z_{h}^{2}]   \leq\frac{1}{n^{2}} \sum_{i\in \tilde{P}}\frac{w_{i}^{2}}{p_{i}}\left(i+\sum_{j>i}\frac{p_{j}}{p_{i}}\right)\label{eq:second-moment}
\end{align}
\end{lemma}
\begin{proof} The proof is based on applying Bayes rule and total probability law. 
\begin{align*}
\E[Z_{h}] &= \frac{1}{n}\sum_{i\in \tilde{P},k\geq 1} \frac{w_{i}}{p_{i}} k \P(\tilde{I}=i, |H(q)|=k)\\
&= \frac{1}{n}\sum_{i\in \tilde{P},k\geq 1} \frac{w_{i}}{p_{i}} k \P(I=i|i\in H(q),|H(q)|=k)  \P(i\in H(q),|H(q)|=k)\\
& = \frac{1}{n}\sum_{i\in \tilde{P}}\frac{w_{i}}{p_{i}}\sum_{k\geq 1} \P(i\in H(q), |H(q)|=k)\\
& = \frac{1}{n}\sum_{i\in \tilde{P}} \frac{w_{i}}{p_{i}}\P(i\in H(q))
\end{align*}	
Similarly, we show for the variance
\begin{align*}
\E[Z_{h}^{2}] &= \frac{1}{n^{2}}\E\left[\left(w_{\tilde{I}} \frac{|H(q)|}{p_{I}} \right)^{2}\right]\\&=\frac{1}{n^{2}} \sum_{i,k}\left(w_{i} \frac{|H(q)|}{p_{i}} \right)^{2}\P(H(q)=k,I=i)\\
&= \frac{1}{n^{2}}\sum_{i,k} \frac{w_{i}^{2}k^{2}}{p_{i}^{2}} \P(H(q)=k, I=i|i\in H(q)) \P(i\in H(q) )\\
&= \frac{1}{n^{2}}\sum_{i,k} \frac{w_{i}^{2}k^{2}}{p_{i}} \P(I=i|H(q)=k, i\in H(q)) \P(H(q)=k|i\in H(q) )\\
&= \frac{1}{n^{2}}\sum_{i,k\geq 1} \frac{w_{i}^{2}}{p_{i}}\sum_{k\geq 1} k \P(H(q)=k|i\in H(q))\\
&= \frac{1}{n^{2}}\sum_{i} \frac{w_{i}^{2}}{p_{i}}\E[|H(q)||i\in H(q)]\\
&= \frac{1}{n^{2}}\sum_{i}  \frac{w_{i}^{2}}{p_{i}} \left( 1+\sum_{j\neq i} \frac{\P(i,j \in H(q))}{p_{i}}\right)
\end{align*}
Finally, we obtain the upper bound by using the fact that $\P(A,B)\leq \min\{\P(A),\P(B)\}$ for any two events.	
\begin{eqnarray}
\E[Z_{h}^{2}] &=&  \frac{1}{n^{2}}\sum_{i}  \frac{w_{i}^{2}}{p_{i}} \left( 1+\sum_{j\neq i} \frac{\P(i,j \in H(q))}{p_{i}}\right)\\
&\leq & \frac{1}{n^{2}}\sum_{i}  \frac{w_{i}^{2}}{p_{i}} \left( 1+\sum_{j\neq i} \frac{\min\{p_{i},p_{j}\}}{p_{i}}\right)\\
&=& \frac{1}{n^{2}}\sum_{i}  \frac{w_{i}^{2}}{p_{i}} \left( i+\frac{1}{p_{i}}\sum_{j>i}p_{j}\right)
\end{eqnarray}
\end{proof}
 Observe that the estimator is unbiased iff $w_{i}>0\Rightarrow p_{i}>0$ for all $i\in[n]$. The quantity in the parenthesis plays the role of $F(w,P)$ and expresses a pessimistic upper bound on  $\E\left[|H(x)|\big|i\in H(x)\right]$. This is perhaps the single most important aspect of the method in that it allows us to \emph{express the variance purely in terms of known collision probabilities}   $p_{i}$ that are  amenable  to analytic manipulations. 

\subsection{Variance Bounds for general  HBE}
In this section, we prove Theorem \ref{thm:two-points} that provides a characterization of the variance of HBE under worst case assumptions. The main tool we employ is the following inequality.

\begin{lemma}\label{lem:quadratic_form}
 Given a positive vector $w\in \R^{n}_{++}$ and a parameter $\mu>0$, define $f_{i}^{*}:=\min\{1,\frac{\mu}{w_{i}}\}$. For any matrix $A\in \R^{n\times n}$:
\begin{equation}
\sup_{\|f\|_{w,1}\leq \mu,\|f\|_{1}\leq 1}\left\{f^{\top}Af\right\}\leq 4\sup_{ij\in[n]}\left\{f_{i}^{*}|A_{ij}|f_{j}^{*}\right\}
\end{equation}
\end{lemma}
\begin{proof}
Let $S_{+}(\mu) =\{i\in [n]|w_{i}\geq \mu\}$ and $S_{-}(\mu)=\{i\in[n]|w_{i}<\mu\}$ and define $A_{\pm\pm} := A_{S_{\pm}(\mu),S_{\pm}(\mu)}$, where $A_{U,V}$ is the matrix with elements $A_{ij}$, $i\in U, j\in V$. We have the decomposition:
\begin{align}
f^{\top}Af = f_{+}^{\top}A_{++}f_{+} + f_{+}^{\top}A_{+-}f_{-}+f_{-}^{\top}A_{-+}f_{+}+f_{-}^{\top}A_{--}f_{-}
\end{align}
Using this decomposition and Lemma \ref{lem:two-sided} we have
\begin{align}
\sup_{\|f\|_{w,1}\leq \mu, \|f\|_{1}\leq 1}\{f^{\top}Af\} &\leq \sup_{\|f_{+}\|_{w,1}\leq \mu, \|f_{+}\|_{w,1}\leq\mu}\{f_{+}^{\top}A_{++}f_{+}\} 
 +\sup_{\|f_{+}\|_{w,1}\leq \mu, \|f_{-}\|_{1}\leq1}\{f_{+}^{\top}A_{+-}f_{-}\}\\
&\qquad+ \sup_{\|f_{-}\|_{1}\leq 1, \|f_{+}\|_{w,1}\leq\mu}\{f_{-}^{\top}A_{-+}f_{+}\} + \sup_{\|f_{-}\|_{1}\leq 1, \|f_{-}\|_{1}\leq1}\{f_{-}^{\top}A_{--}f_{-}\}\\
&\leq \sup_{i\in S_{+},j\in S_{+}}\left\{\frac{\mu}{w_{i}}|A_{ij}|\frac{\mu}{w_{j}}\right\}+ \sup_{i\in S_{+},j\in S_{-}}\left\{\frac{\mu}{w_{i}}|A_{ij}|\frac{1}{1}\right\}\\
&\qquad    + \sup_{i\in S_{-},j\in S_{+}}\left\{\frac{1}{1}|A_{ij}|\frac{\mu}{w_{j}}\right\}  + \sup_{i\in S_{-},j\in S_{-}}\left\{\frac{1}{1}|A_{ij}|\frac{1}{1}\right\}
\end{align}
To complete the proof note that $\min\{1,\frac{\mu}{w_{i}}\} = \frac{\mu}{w_{i}} \forall i \in S_{+}(\mu)$ and $\min\{1,\frac{\mu}{w_{j}}\} = 1 \forall j\in S_{-}(\mu)$. 
\end{proof}
\begin{proof}[Proof of Theorem \ref{thm:two-points}] Fix $n$ distinct weights $w_{1},\ldots, w_{n}$ and collision probabilities $p_{1},\ldots, p_{n}$ that might be considered. Without loss of generality we assume that the probabilities are in decreasing order. Let $f_{i}$ be the fraction of points that are assigned to weight $i$. We have the following:
\begin{align*}
\E[Z_{h}^{2}] &\leq \frac{1}{n^{2}} \sum_{i}(f_{i}n)  \frac{w_{i}^{2}}{p_{i}} \left( n\sum_{j\leq i} f_{j}+\frac{1}{p_{i}}\sum_{j>i}p_{j}(nf_{j})\right)\\
& = \sum_{i,j}f_{i}f_{j} \left(\frac{w_{i}^{2}}{p_{i}} \I_{j\leq i} + \I_{j>i}\frac{w_{i}^{2}}{p_{i}^{2}}p_{j}\right)\\
&\leq  \sup_{\tiny \begin{matrix}
\|f\|_{w,1}\leq \mu\\
\|f\|_{1}\leq 1
\end{matrix}}\left\{f^{\top}Af\right\}
\end{align*}
 where in the last step we set $A_{ij}:=\frac{w_{i}^{2}}{p_{i}} \I_{j\leq i} + \I_{j>i}\frac{w_{i}^{2}}{p_{i}^{2}}p_{j}$. Invoking Lemma \ref{lem:quadratic_form} shows that for any set of weights $\{w_{i}\}_{i\in[n]}$ and probabilities $\{p_{i}\}_{i\in[n]}$ there exist two indices $i^{*},j^{*}\in [n]$ that realize up to a constant the worst case variance for the specific density $\mu$.
\end{proof}
The theorem assumes nothing about HBE besides unbiasedness and therefore is applicable in an arbitrary setting. If more structure is assumed stronger statements can be made. 
\subsection{Variance Bounds for Scale-free HBE}
We   give a refined analysis  of \emph{scale-free} estimators (Definition \ref{def:scale-free}) based on whether a large fraction of the density comes from points with relatively large weights. 

\begin{definition} For a query $x$ and $\tau \in[\mu, 1]$ let $B_{\tau,\mu}(x):=\{ i\in P|w_{i} \geq \frac{\mu}{\tau}\}$. For $\gamma\in[0,1]$ the query is said to be $(\tau,\gamma)$-\emph{localized}  if $\sum_{i\in B_{\tau,\mu}} w_{i} \geq (1-\gamma) \sum_{i} w_{i}$.
\end{definition}
The intuition behind this definition is that for kernels that are decreasing with the distance, the sets of points that would induce large (small) weights are contained in (outside of)  a ball of a certain radius around the query. If the total contribution of the points outside of the ball is negligible, our estimator behaves essentially as if we are at a ``higher" density regime.

\begin{theorem}[Restatement of Theorem \ref{thm:scale-free-var}]\label{thm:scale-free}
 Let $Z_{h}$ be a $(\beta,M)$-\emph{scale free} estimator with $\beta \in [1/2,1]$. For every $(\tau,\gamma)$-localized query $x$, 
 \begin{align*}
\E[Z_{h}^{2}] &\leq \mu^{2}\cdot M^{3}\left\{2\tau^{\beta}+\gamma^{2-\beta}+ \tau^{2\beta -1}\gamma^{\beta}\right\}   \mu^{-\beta}
 \end{align*}
\end{theorem}
 The optimal choice of $\beta$ in the case where no assumptions are made $(\gamma=\tau=1)$  is $\beta^{*}=\frac{1}{2}$.  Even if $\beta$ is not selected depending on the structure of the query, significant gains can be obtained, in cases where the parameters $\tau,\gamma$ are small enough. As an example, for $\beta =\frac{1}{2}$, $\tau =\mu^{\frac{1}{2}}$, $\gamma=\mu^{\delta}$ the theorem implies that the variance is bounded by $\mu^{2}\cdot 4M^{3}\mu^{-\Delta(\delta)}$ where $\Delta(\delta)=\max\{\frac{1}{4}-\delta,\frac{1-3\delta}{2}\}<\frac{1}{2}$ for all $\delta>0$. 

\begin{proof} 
We start by using the scale-free property to simplify \eqref{eq:second-moment}:
\begin{equation}
\E[Z_{h}^{2}]\leq M^{3}\cdot\frac{1}{n^{2}} \underbrace{\sum_{i=1}^{n}w_{i}^{2-\beta}\left(i+\sum_{j>i}\left(\frac{w_{j}}{w_{i}}\right)^{\beta} \right)}_{A}
\end{equation}
To bound $A$, we break up the terms according to $B_{\tau,\mu}$. Let $J$ be the maximal index in $B_{\tau,\mu}$.
\begin{align}
A & = \sum_{i\leq J}w_{i}^{2-\beta}\left(i+\sum_{J\geq j>i}\left(\frac{w_{j}}{w_{i}}\right)^{\beta} \right) + \sum_{i>J}w_{i}^{2-\beta}\left(i-J+\sum_{j>i}\left(\frac{w_{j}}{w_{i}}\right)^{\beta} \right)\\
&\qquad +  \sum_{i\leq J}w_{i}^{2-2\beta} \cdot \sum_{j>J}w_{j}^{\beta} + |B_{\tau,\mu}|\sum_{i> J}w_{i}^{2-\beta} 
\end{align}
We next use H\"{o}lder-type inequalities to bound each term. For the first two terms we invoke Lemma \ref{lem:mon-holder} (Monotone H\"{o}lder) to get:
\begin{align}
\sum_{i\leq J}w_{i}^{2-\beta}\left(i+\sum_{J\geq j>i}\left(\frac{w_{j}}{w_{i}}\right)^{\beta} \right)&\leq (n\tau)^{\beta} (n\mu)^{2-\beta} \\
\sum_{i>J}w_{i}^{2-\beta}\left(i-J+\sum_{j>i}\left(\frac{w_{j}}{w_{i}}\right)^{\beta} \right) &\leq   n^{\beta} (\gamma n \mu)^{2-\beta}
\end{align}
where we also used the fact that $|B_{\tau,\mu}|\leq \tau n$. 
Let $w_{\leq} = (w_{1},\ldots, w_{J})$ and $w_{>}:=(w_{J+1},\ldots, w_{n})$.
To bound the remaining terms, we invoke Corollary \ref{col:holder} in total three times to get:
\begin{align}
\sum_{i\leq J}w_{i}^{2-2\beta} \cdot \sum_{j>J}w_{j}^{\beta}  &\leq \|w_{\leq}\|_{1}^{2-2\beta} |B_{\tau,\mu}|^{1-(2-2\beta)} \cdot \|w_{>}\|_{1}^{\beta}n^{1-\beta}\\
 &\leq (n\mu)^{2-2\beta} (\tau n)^{1-(2-2\beta)}\cdot (\gamma n\mu)^{\beta}n^{1-\beta}
\end{align}
and
\begin{align}
 |B_{\tau,\mu}|\sum_{i> J}w_{i}^{2-\beta} &\leq n\tau\cdot  \|w_{>}\|_{\infty}^{1-\beta} \cdot \|w_{>}\|^{1}_{1}  \ \leq n \tau \left(\frac{\mu}{\tau}\right)^{1-\beta} (\gamma n\mu)
\end{align}
By combining all the inequalities, we arrive at $A \leq n^{2}\mu^{2-\beta}\cdot \left\{\tau^{\beta}+\gamma^{2-\beta}+ \tau^{2\beta -1}\gamma^{\beta}+\tau^{\beta}\gamma\right\}$.
\end{proof}

 This concludes the presentation of the general framework of Hashing-Based-Estimators, that given a set of weights $\{w_{i}\}$ and collisions probabilities $\{p_{i}\}$ computes a function $V$ such that the resulting HBE is $V$-bounded.
\section{Adaptive Estimation through Mean Relaxation}\label{sec:amr}
Our goal in this section is, given a dataset $P\in \R^{d}$ and a $V$-bounded  estimator, to build a data structure that can efficiently approximate the mean $\mu(x)$ for a given query $x\in \R^{d}$ . Thus, presenting a complete algorithmic framework that can be instantiated for different problems.

We first address the problem of obtaining a constant factor approximation to the mean. We exploit two facts:   monotonicity of the variance in terms of $\mu$ and concentration of measure. Monotonicity suggests that we can start with an over-estimate of the mean and keep refining it until we come very close to the truth. Concentration of  measure allows us to come up with a simple consistency check that recognizes when our estimate is close enough to the true mean. Based on this we propose the following adaptive algorithm for which we get strong guarantees.

\begin{algorithm}[h]
	\caption{Adaptive Mean Relaxation (AMR)}
	\label{alg:example}
	\begin{algorithmic}[1]
		\STATE {\bfseries Input:}   V-bounded unbiased estimator    $Z\sim \nu$, query $x\in \R^{d}$, accuracy $\alpha\in (0,1]$, threshold $\tau \in (0,1)$, failure prob. $\chi\in (0,1)$.
		\STATE $\epsilon\gets \frac{2}{7}\alpha, c\gets \frac{\epsilon}{2}, \gamma\gets \frac{\epsilon}{7}$, $\delta\gets \frac{2\alpha}{49\log(1/\tau)}\chi$, $i \gets -1$.
		\REPEAT  
		\STATE $i\gets i+1$, $\mu_{i}\gets(1-\gamma)^{i}$,  
		\STATE $Z_{i} \gets \mathrm{MoM}_{\frac{\epsilon}{3},\delta}(\nu,V(\mu_{i}))$\label{code:MoM}
		\UNTIL {$|Z_{i}-\mu_{i}|\leq c\cdot \mu_{i}$ or $i>\frac{\log(\tau/(1-(c+\epsilon)))}{\log(1-\gamma)}$.}
		\STATE {\bf Output:} if $i\leq \frac{7^{2}\log( 1/\tau)}{2\alpha }$ return $Z_{i}$ else return $0$.
	\end{algorithmic}
\end{algorithm}

In fact, there is no need for the samples used in different calls of the MoM routine to be independent and we can implement the MoM routing by keeping $L(\delta)=\lceil 9\log(1/\delta)\rceil$ running sums. We call the resulting algorithm AMR*.
\begin{theorem}[Mean Relaxation]\label{thm:AMR}
 Let $Z\sim\nu$ be a $V$-bounded estimator with $\E[Z] = \mu \in (0,1]$ and $\hat{Z}$ be the output of the AMR* algorithm with parameters $(\alpha, \tau, \chi)$. If  $\mu \geq \tau$ then $
\P[|\hat{Z}-\mu|\leq \alpha \mu] \geq 1 -\chi 
$ otherwise if $\mu <\tau$,  $\P[ \hat{Z}=0 ]\geq 1- \chi$. The total number of   samples used is bounded  by 
 $O\left(\alpha^{-2} \log(\frac{\log(1/\tau)}{\chi\alpha})\cdot   V((\mu)_{\tau})\right)$.
\end{theorem}

With this algorithm in hand we are ready to state our main result in the abstract setting of HBE. 
\begin{theorem}[Main Result]\label{thm:main} 
Given a $V$-bounded HBE with complexity $T$ , there exists a data structure that can answer any  query in time $ O(\frac{1}{\epsilon^{2}}V((\mu)_{\tau})\log(\frac{1}{\chi})  T)$ using space $O(\frac{1}{\epsilon^{2}}V(\tau)\log(\frac{1}{\chi})\cdot nT)$ with success probability at least $1-\chi$ for $\chi\leq 1/\log(1/\tau)$.
\end{theorem}
\begin{proof}
We begin by describing the preprocessing phase. We sample $N= O(\log(1/\chi) \frac{1}{\epsilon^{2}} V(\tau))$ hash functions $h_{1},\ldots, h_{N}\stackrel{i.i.d.}{\sim} \nu$ from $\mathcal{H}$ and evaluate them on the dataset $P$. This can be done in  $NT\cdot n$ time and space. The query algorithm interacts with the data-structure by making calls to hash functions. The data-structure always keeps the index of the last hash function called and increments it in a cyclic fashion after each call, thus for a given query it never evaluates the same hash function twice and the samples obtained are independent.  When the query arrives, the query algorithm first  runs a stage of the adaptive mean relaxation algorithm with $\alpha=1$ and probability $\chi/2$. Every time a sample is needed a call is made to the data-structure. After $O(V((\mu)_{\tau}\log(1/\chi)))$ calls with probability at least $1-\chi/2$ we either have a constant factor approximation or we know that $\mu<\tau$ (if AMR* outputs $0$). In the first case, we apply one level of MoM algorithm using an underestimate of $\mu$ that uses $O(\frac{1}{\epsilon^{2}}\log(1/\chi)V((\mu)_{\tau}))$ more calls and gets a $(1\pm \epsilon)$ multiplicative  approximation with probability at least $1-\chi$. In the latter, case we simply output 0.  
\end{proof}

The proof of Theorem \ref{thm:main-kernel} follows by invoking Theorem \ref{thm:main} for specific $V$-bounded estimators that we derive in Sections \ref{sec:euclidean} and \ref{sec:gaussian} using Locality Sensitive Hashing schemes. In the rest of this section we give the proof of the intermediate results. We use the following key lemma.

\begin{lemma}\label{lem:adaptive-mom} Given a $V$-bounded estimator $Z\sim \nu$,  let $Z_{i}:=\mathrm{MoM}_{\frac{\epsilon}{3}}(\nu,V(\mu_{i}))$, with $\mu_{i}\in(0,1]$ then
\begin{enumerate}
\item [(i)] for  all $i$ such that $\mu_{i} \geq (1-(c+\epsilon))^{-1} \mu $, it holds that $\P[|Z_{i}-\mu_{i}|\leq c\cdot \mu_{i}]\leq \delta$.
\item [(ii)] for all $\frac{1+\frac{\epsilon}{3}}{1+c}\mu\leq \mu_{i} \leq\mu$, it holds that $\P[|Z_{i}-\mu_{i}|\leq c \cdot \mu_{i}] \geq 1- \delta$.
\end{enumerate}
\end{lemma}

\subsection{Proof of Theorem \ref{thm:AMR}}
 
\subsubsection{Case $\mu \geq \tau$}
If $\mu \geq \tau$, we say that the algorithm succeeds when the algorithm terminates with $\frac{1+\epsilon/3}{1+c}\mu\leq \mu_{i}\leq \frac{1}{1-(c+\epsilon)}\mu$. By construction of the algorithm this would imply $|Z_{i}-\mu|\leq \alpha \mu$. To show this we consider two cases. The first case is when $\mu_{i}$ is in the subinterval with elements larger than $\mu$:
\begin{align}
\left.
\begin{matrix}
\mu \leq \mu_{i}\leq (1-(c+\epsilon))^{-1}\mu\\
|Z-\mu_{i}|\leq c\mu_{i}
\end{matrix}\right\} \Rightarrow |Z-\mu_{i}|\leq \left( (c+\epsilon)+\frac{c}{1-(c+\epsilon)}\right)\mu \leq \alpha \mu
\end{align}
where in the last step we used that $\epsilon=2c=2\alpha/7$ and $1-3\alpha/7\geq 1/2$. The other case is when:
\begin{align}
\left.
\begin{matrix}
\frac{1+\epsilon/3}{1+c}\mu \leq \mu_{i}\leq \mu\\
|Z-\mu_{i}|\leq c\mu_{i}
\end{matrix}\right\} \Rightarrow |Z-\mu_{i}|\leq  \frac{2\epsilon}{3}\mu \leq \alpha \mu
\end{align}
Thus, what is left is to bound the probability that the algorithm terminates for some $i$ that satisfies $\frac{1+\epsilon/3}{1+c}\mu \leq \mu_{i}\leq (1-(c+\epsilon))^{-1}\mu$.

\begin{claim}
There exists  $0\leq i^{*}\leq \frac{\log(\mu /(1-(c+\epsilon)))}{\log(1-\gamma)}$ such that $\mu \leq \mu_{i^{*}}\leq (1-(c+\epsilon))^{-1}\mu$.
\end{claim}
\begin{proof}
If $\mu> (1-(c+\epsilon))$ then this is true for $i^{*}=0$, otherwise there exist  $i \geq 1$ such that  $\mu_{i -1}\geq \frac{1}{1-(c+\epsilon)}\mu$. Since we start with $\mu\leq \mu_{0}$, we  show that by slowly decreasing $\mu_{i}$ (at a rate of $1-\gamma$) we never go from the case $\mu_{i-1}\geq (1-(c+\epsilon))^{-1}\mu$ to the case $\mu_{i}<\mu$. Assuming that was possible  we obtain a contradiction. 
\begin{align}
(1-\gamma)=\frac{\mu_{i}}{\mu_{i-1}}< \frac{\mu}{(1-(c+\epsilon))^{-1}\mu} = 1-(c+\epsilon) < 1-\gamma
\end{align}
Since the sequence $\mu_{i}=(1-\gamma)^{i}$ there must be an index $i\leq \log(\mu/(1-c-\epsilon))/\log(1-\gamma)$ with the required property.
\end{proof}

\begin{claim} There exists $j^{*}>i^{*}$ such that $\frac{1+\epsilon/3}{1+c}\mu \leq \mu_{j^{*}}\leq \mu$.
\end{claim}
\begin{proof} Let $i^{*}$ be as above, then  $\mu \leq \mu_{i^{*}}$. We  show that for $j>i^{*}$ we cannot have $\mu \leq \mu_{j-1}$ and $\mu_{j} < \frac{1+\epsilon/3}{1+c}\mu$ as in that case $1-\gamma=\frac{\mu_{j}}{\mu_{j-1}}<\frac{1+\epsilon/3}{1+c}\Rightarrow \gamma >\frac{\epsilon}{6(1+c)}>\frac{\epsilon}{7}=\gamma$. As the sequence $\mu_{i}$ is decreasing there must exist a $j^{*}$ with the required property and $j^{*}\leq i^{*} + \lceil \frac{\log(\frac{1+c}{(1-(c+\epsilon))(1+\epsilon/3)})}{-\log(1-\gamma)}\rceil=O(i^{*})$
\end{proof}
The algorithm thus succeeds if the condition $|Z_{i}-\mu_{i}|<c\mu_{i}$ is satisfied for any index $k \in \{i^{*}, \ldots, j^{*}\}$ and is not satisfied for any index $k < i^{*}$. Let $A_{i}$ be the event that condition is satisfied at step $i\geq 0$. Lemma \ref{lem:adaptive-mom} shows that $\P[A_{i}]\leq \delta$ for all $i<i^{*}$ and $\P[A_{j^{*}}]\geq 1-\delta$. The probability of success is then at least:
\begin{align}
\P[\bar{A}_{0}\cap \ldots \cap \bar{A}_{i^{*}-1}\cap A_{j^{*}}] = 1- \P[A_{0}\cup A_{i^{*}-1}\cup \bar{A}_{j^{*}}]\geq 1- (i^{*}+1)\delta \geq 1-\chi
\end{align} 
The number of samples required for the $\textrm{AMR}^{*}$ procedure are given by:
\begin{equation}
\lceil \frac{54}{\epsilon^{2}}V(\mu_{j^{*}})\rceil \lceil 9\log(1/\delta)\rceil = O\left(\frac{1}{\alpha^{2}}V((\mu)_{\tau})\cdot \log\left(\frac{\log(1/\tau)}{\chi\alpha}\right)\right)
\end{equation} 
\subsubsection{Case $\mu<\tau$}
When $\mu < \tau$ we are always in the case where  $\mu_{i} \geq (1-(c+\epsilon))^{-1} \mu $ and by Lemma \ref{lem:adaptive-mom} the probability that the algorithm terminates before step $\log(\tau/(1-(c+\epsilon)))/\log(1-\gamma)$ is upper bounded by $\log(\tau/(1-(c+\epsilon)))/\log(1-\gamma)\delta\leq \chi$. Thus, with probability at least $1-\chi$ we have $\hat{Z}=0$.
 
\subsection{Proof of Lemma \ref{lem:adaptive-mom}}

To prove the lemma, we are going to show that the event $|Z_{i}-\mu_{i}| \leq c \mu_{i}$ is related to an event that expresses the deviation of $Z_{i}$ from its true mean. We will then appeal to concentration of measure for the median-of-means estimator to show the required bounds. We have that for all $\mu_{i} \geq (1-(c+\epsilon))^{-1} \mu $ 
\begin{align*}
|Z_{i}-\mu_{i}|\leq c   \mu_{i} 
&\Rightarrow |Z_{i}-\mu| \geq \left[\left(1-\frac{\mu}{\mu_{i}} \right) -c\right] \mu_{i}  \Rightarrow |Z_{i}-\mu|\geq \epsilon \mu_{i}
\end{align*} 
By monotonicity this implies that $\P[|Z_{i}-\mu_{i}|\leq c \mu_{i}] \leq \P\left[|Z_{i}-\mu|\geq \epsilon \mu_{i}  \right]$. But $Z_{i}=\mathrm{MoM}_{\frac{\epsilon}{3},\delta}(\nu,V(\mu_{i}))$, is the median of $L(\delta)=\lceil 9\log(1/\delta)\rceil $, averages of $54V(\mu_{i})/\epsilon^{2}$ samples. We compute the probability that any of the $L$ averages deviates from the true mean:
\begin{equation*}
\P[|\tilde{Z}-\mu|\geq \epsilon \mu_{i}] \leq \frac{\mu^{2}\cdot V(\mu)}{\epsilon^{2}\mu_{i}^{2} \frac{54V(\mu_{i})}{\epsilon^{2}}} \leq \frac{1}{54} \cdot\frac{\mu^{2}V(\mu)}{\mu_{i}^{2}V(\mu_{i})} \leq \frac{1}{54}
\end{equation*}
where the last step follows by monontonicity of $\mu^{2}V(\mu)$ and the assumption that $\mu_{i}\geq (1-(c+\epsilon))^{-1}\mu$. Then, by Chernoff bounds (Appendix ) we get that the probability that the median  of such averages has the deviation in question is less than $\delta$. Similarly, we have that in the case that $\frac{1+\frac{\epsilon}{3}}{1+c}\mu\leq \mu_{i} \leq \mu$:
\begin{align}
|Z_{i}-\mu | \leq \frac{\epsilon}{3} \mu \Rightarrow |Z_{i}-\mu_{i}|\leq  c   \mu_{i}    
\end{align} 
As before let $\tilde{Z}$ denote any of the $L$ averages of which we output the median. Then,
\begin{align*}
\P[|\tilde{Z}-\mu_{i}|\leq c \mu_{i}] \geq \P[|\tilde{Z}-\mu|\leq \frac{\epsilon}{3}\mu]\geq 1-\P[|\tilde{Z}-\mu|\geq \frac{\epsilon}{3}\mu]\geq  1- \frac{9\mu^{2} V(\mu)}{\epsilon^{2}\mu^{2}\frac{54V(\mu_{i})}{\epsilon^{2}}} \geq 1-\frac{1}{6} \frac{V(\mu)}{V(\mu_{i})} \geq \frac{5}{6}
\end{align*}
Again, by Chernoff bounds we get that the probability of the deviation for the median is less than $\delta$.

\section{KDE through Euclidean LSH}\label{sec:euclidean}
In this section, we instantiate the framework of HBE for the problem of KDE using Euclidean LSH of Datar et al.~\cite{datar2004locality}. At a high level, given a kernel $k$, the goal is to design a hashing scheme such that the probability that two points hash at the same bucket is as similar as possible to $k(x,y)$. We consider three such kernels: the Exponential, the Generalized $t$-Student (polynomial) and the Gaussian kernel. We show that this specific LSH scheme can be used to construct scale-free estimators for the first two, while for the Gaussian case, although we cannot get a scale-free estimator,  we are still able to analyze the variance of the estimator using Theorem \ref{thm:two-points}. 
\subsection{Euclidean LSH}
The family of hash functions is given by:
\begin{align*}
\mathcal{H}_{1}(w)&:=\left\{ \left.h(x)=\left\lceil \frac{g^{\top}x+\beta}{w}\right\rceil \right| g\in\R^{d}, \beta \in [0,w] \right\}
\end{align*}
for some fixed $w>0$. 
We define a distribution $\nu_{1}$ over $\mathcal{H}_{1}$ by sampling $g\sim \mathcal{N}(0,I_{d})$ and $\beta \sim U[0,w]$. The important quantity to control is the collision probability $p_{1}(c):=\P_{h\sim \nu}\left[h(x)=h(y)\right]$ of two points $x,y$ at distance $\|x-y\|=cw$ and is given by~\cite{datar2004locality}:
\begin{equation}\label{eq:DIIM}
p_{1}(c)= 1 - 2\Phi(c^{-1}) - \sqrt{\frac{2}{\pi}} c \left(1-\exp\{-\frac{c^{-2}}{2}\}\right)
\end{equation}
\begin{lemma}[Pointwise bounds]\label{lem:one-dim}
For all $c>0$
\begin{equation}
p_{1}(c) = \sqrt{\frac{2}{\pi}}\sum_{k=0}^{\infty} \frac{(-1)^{k}}{2^{k}k!(2k+2)(2k+1)}\frac{1}{c^{2k+1}}
\end{equation}

while for  $\delta\leq \frac{1}{2}$ and $c\leq \min\left\{\delta,\frac{1}{\sqrt{2\ln(1/\delta)}}\right\}$ 
\begin{align} \label{eq:exp-bounds}
e^{-\sqrt{\frac{2}{\pi}} \left(1+\delta \right)\cdot c}\leq p_{1}(c) \leq e^{-\sqrt{\frac{2}{\pi}} (1-\delta^{3})\cdot c}
\end{align}
\end{lemma}
We see that this family of hash functions gives collision probability that is  exponentially decreasing with the distance for small values of $c\ll 1$ whereas the decay is inverse polynomial for large values values of $c\gg 1$. We exploit these  properties to derive scale free estimators for the Exponential and $t$-Student kernels. 

\subsection{Scale-free Estimators for Exponential and $t$-Student Kernel}

\begin{theorem}[Exponential Kernel]\label{thm:exponential}
  For $\beta \in(0,1]$ there exists a $(\beta,\sqrt{e})$-scale free HBE for the exponential kernel $ e^{-\|x-y\|}$ that has complexity $O(d R^{2})$.
\end{theorem}
\begin{proof}
Set $D=\lceil \sqrt{2\pi} R \rceil^{2}, \ w= \frac{D}{\beta \sqrt{\frac{\pi}{2}}}$ and consider the HBE resulting from the family $\mathcal{H}_{1}^{\otimes D}(w)$ with probability measure $\nu_{1}^{\otimes D}$. We first see that for a pair of points that are distance $\|x-y\|=r$ apart, we have that $c\leq \delta = \frac{R}{w} = \beta\sqrt{\frac{\pi}{2}} \frac{R}{\lceil\sqrt{2\pi} R\rceil^{2} } \leq \frac{\beta}{2}\frac{1}{\lceil\sqrt{2\pi} R\rceil} \leq \frac{1}{2}$. Additionally, for these parameters it holds that $e^{-\sqrt{\frac{2}{\pi}}cD} = e^{-\beta \cdot r}$. Using the second part of Lemma \ref{lem:one-dim} with $\delta = \beta\sqrt{\frac{\pi}{2}} \frac{R}{\lceil\sqrt{2\pi} R\rceil^{2} }$ we get 
\begin{align}
e^{-\sqrt{\frac{2}{\pi}}(\frac{\beta}{2}\frac{1}{\lceil\sqrt{2\pi} R\rceil})^{2}D}   \leq  \frac{p_{1}^{D}(c)}{e^{-\beta \cdot r}} \leq e^{+\sqrt{\frac{2}{\pi}}(\frac{\beta}{2}\frac{1}{\lceil\sqrt{2\pi} R\rceil})^{4}D} 
\end{align}
Which gives us that $\frac{1}{\sqrt{e}} \cdot e^{-\beta   r} \leq p^{D}_{1}(c) \leq \sqrt{e}\cdot e^{-\beta  r}$. To implement the estimator we require $O(R^{2})$ hash functions that each can be evaluate in time $O(d)$ and requires space $O(d+n)$. 
\end{proof}

Although, this is to be expected given the derived pointwise bounds on the collision probabilities, and seems to suggest a correspondence between kernels and hashing schemes, we show next that one can construct a \emph{scale-free estimator for another very different kernel using the same LSH scheme}.
\begin{theorem}[Generalized $t$-Student Kernel]\label{thm:student}
  For integers $p,q\geq 1$, there exists a $(\frac{q}{p},3^{q})$-scale free HBE for the kernel $\frac{1}{1+\|x-y\|^{p}}$ that has complexity $T=O(dp)$.
\end{theorem}

\begin{proof} Set $w=\sqrt{2\pi}$ and consider the HBE resulting from the family $\mathcal{H}_{1}^{\otimes p}(w)$. We obtain bounds for the collision probabilities for two points at distance $\|x-y\|=r =c\cdot w$ apart.  For $c> 1$, using the first part of Lemma \ref{lem:one-dim} and dropping terms appropriately 
\begin{equation}\label{eq:poly-bound}
\frac{1}{\sqrt{2\pi}} (1-\frac{1}{12c^{2}})   \frac{1}{c}\leq p_{1}(c) \leq \frac{1}{\sqrt{2\pi}} \frac{1}{c}
\end{equation}

Setting $\beta = \frac{q}{p}$ and raising to the $q$-th power we get that for all $c> 1$:
\begin{align}
\left( \frac{11}{12}\right)^{q}   \frac{(1+r^{p})^{\beta}}{r^{q}} \leq \frac{p_{1}^{q}(c)}{\frac{1}{(1+r^{p})^{\beta}}} \leq  \frac{(1+r^{p})^{\beta}}{r^{q}}
\end{align}

whereas for $c\leq 1$ we have from \eqref{eq:poly-bound} and monotonicity of $p_{1}(c)$:
\begin{equation}
\left(\frac{11}{12\sqrt{2\pi}}\right)^{q}(1+r^{p })^{\beta}\leq \frac{p_{1}^{q}(c)}{\frac{1}{(1+r^{p})^{\beta}}}\leq (1+r^{p})^{\beta}
\end{equation}
Since $2^{-1} \geq 11/(12\sqrt{2\pi}) \geq 3^{-1}$ this shows that the resulting estimator is $(\frac{q}{p}, 3^{q})$-\emph{scale-free} for the Generalized $t$-student Kernel for all distances.
\end{proof}
These two results combined with Theorem \ref{thm:scale-free} and Theorem \ref{thm:main} can be used to construct a data structure for the KDE problem under the exponential or $t$-Student kernel. More importantly, this is done effortlessly by appealing to the the general case of HBE and only required showing that our hashing schemes produces scale-free estimators. 

\subsection{An estimator for Gaussian KDE}
We show next that the framework of HBE can be useful beyond the ideal scenarios where a scale-free estimator can be derived. We do so by showing that one can use the exponential drop-off of the collision probabilities to simulate the Gaussian kernel and then appeal to the more general Theorem \ref{thm:two-points} to bound its variance.

\begin{theorem}\label{thm:euclid-gauss} 
For any $ t\in [1, R]$ there exists a HBE $Z_{t}$ for the Gaussian kernel $e^{-\|x-y\|^{2}}$ with $\E[Z_{t}^{2}] \leq \mu^{2}\cdot4e^{\frac{3}{2}} \mu^{-\gamma^{2}+\gamma-1}$ where $\gamma(t,\mu) := t/\sqrt{\log(1/\mu)}$, that has complexity $T=O(d t^{2}R^{2})$.
\end{theorem}

\begin{proof}
Fix $t\geq 1$ and set $D= 3 \lceil tR\rceil^{2}$, $ w= \frac{D}{t} \sqrt{\frac{2}{\pi}}$. We consider once more the HBE resulting from $\mathcal{H}_{1}^{D}(w)$ and probability measure $v_{1}^{\otimes D}$. For a pair of points at distance $r = c\cdot w\leq R$ we have that $c\leq \delta = \frac{R}{w} \leq  \frac{tR}{  3 \lceil tR\rceil^{2} }\sqrt{\frac{\pi}{2}}\leq \frac{1}{2}$. By utilizing Lemma \ref{lem:one-dim} one again we obtain   $\frac{1}{\sqrt{e}}\cdot  e^{- r\cdot t}\leq p^{D}_{1}(c) \leq \sqrt{e} \cdot  e^{- r\cdot t}$. To bound the variance due to Theorem \ref{thm:two-points} we only need to consider what happens for datasets supported only on two points. It will be useful to consider that one point is at distance $r_{1}=\sqrt{\alpha \log(1/\mu)}$ away from the query and the other at  $r_{2}=\sqrt{\alpha^{'}\log(1/\mu)}$ with $0\leq \alpha \leq \alpha^{'}$. Setting $\gamma(t,\mu) = \frac{t}{\sqrt{\log(1/\mu)}}$ the pointwise bounds become
\begin{align}
\frac{1}{\sqrt{e}}\cdot \mu^{\gamma\sqrt{\alpha}} \leq p^{D}_{1}(c) \leq \sqrt{e}\cdot \mu^{\gamma\sqrt{\alpha}}
\end{align}
The weights (contribution) of each point in this case are given by $w_{1}=\mu^{\alpha}$ and $w_{2}=\mu^{\alpha^{'}}$. To bound the variance we pefrom a case analysis:
\begin{itemize}
\item \emph{Case} $\alpha\leq \alpha^{'}\leq 1$: we first bound the term
\begin{equation}
\tilde{A}_{1,2}=\frac{\mu}{w_{1}}A_{1,2}\frac{\mu}{w_{2}}=\mu^{2}\frac{w_{1}}{p_{1}^{2}}\frac{p_{2}}{w_{2}}\leq \mu^{2}e^{3/2}\cdot \mu^{\alpha-2\gamma\sqrt{\alpha}-\alpha^{'}+\sqrt{\alpha^{'}}\gamma}\leq \mu^{2}e^{3/2}\cdot \mu^{-\gamma^{2}+\gamma-1}
\end{equation}
where in the last step we optimized for $0\leq \alpha\leq \alpha^{'}\leq 1$ getting values $\alpha^{'}=1$ and $\alpha=\gamma^{2}$. 
\begin{equation}
\tilde{A}_{2,1}=\frac{\mu}{w_{2}}A_{2,1}\frac{\mu}{w_{2}} = \mu^{2}\frac{w_{2}}{p_{2}w_{1}}\leq \mu^{2} \sqrt{e}\cdot \mu^{\alpha^{'}-\gamma\sqrt{\alpha^{'}}-\alpha}\leq \mu^{2}\sqrt{e}\cdot \mu^{-\gamma}
\end{equation}
where in the last step we set $\alpha=\alpha^{'}=1$. The same bound applies to the term $\tilde{A}_{2,2}$. For completeness, we also bound:
\begin{eqnarray}
\tilde{A}_{1,1}=\left(\frac{\mu}{w_{1}}\right)^{2}A_{1,1}=\frac{1}{p_{1}}\mu^{2}\leq  \mu^{2}\sqrt{e}\cdot \mu^{-\gamma}
\end{eqnarray}
Hence, the variance is bounded by $\mu^{2}4e^{3/2}\cdot \mu^{-\gamma^{2}+\gamma-1}$.
\item \emph{Case} $\alpha\leq 1< \alpha^{'}$: we only need to bound the cross-terms
\begin{equation}
\tilde{A}_{1,2}=\frac{\mu}{w_{1}}\cdot A_{1,2}\cdot 1=\mu \frac{w_{1}p_{2}}{p_{1}^{2}}\leq \mu^{2}e^{3/2}\cdot \mu^{\alpha-2\gamma\sqrt{\alpha}+\sqrt{\alpha^{'}}\gamma-1}< \mu^{2}e^{3/2}\cdot \mu^{-\gamma^{2}+\gamma-1}
\end{equation}
where again we optimized for $\alpha\leq 1 < \alpha^{'}$. The other term is
\begin{equation}
\tilde{A}_{2,1}=1\cdot A_{2,1}\frac{\mu}{w_{1}} = \mu \frac{w^{2}_{2}}{p_{2}w_{1}}\leq \mu  \sqrt{e}\cdot \mu^{2\alpha^{'}-\gamma\sqrt{\alpha^{'}}-\alpha}\leq \mu^{2}\sqrt{e}\cdot \mu^{-\gamma}
\end{equation}
where in the last step we set $\alpha=\alpha^{'}=1$. Additionally,
\begin{eqnarray}
\tilde{A}_{1,1}=\left(\frac{\mu}{w_{1}}\right)^{2}A_{1,1}=\frac{1}{p_{1}}\mu^{2}\leq  \mu^{2}\sqrt{e}\cdot \mu^{-\gamma}
\end{eqnarray}
Hence, in this case the variance is bounded by $\mu^{2}4e^{3/2}\cdot \mu^{-\gamma^{2}+\gamma-1}$.
\item \emph{Case} $1<\alpha\leq \alpha^{'}$: following the same logic we get that the variance is bounded by $\mu^{2}4e^{3/2}\mu^{-\gamma}$.
\end{itemize}

Thus, overall we see that the worst-case distances of the two points are given by  $\alpha=\gamma^{2}$ and $\alpha^{'}=1$. For which, the upper bound on the second moment becomes $\mu^{2}\cdot 4e^{\frac{3}{2}}\cdot \mu^{-\gamma^{2}+\gamma-1}$. 
\end{proof}
Even though that we were not able to get the ideal behavior using this hashing scheme for the Gaussian Kernel, still the estimator has improved variance compared to random sampling for all $0<\gamma<1$ and achieves it's best performance when $t = \frac{1}{2}\sqrt{\log(1/\mu)}\Rightarrow \gamma=\frac{1}{2}$.

\subsection{Pointwise Bounds for Euclidean LSH} 
\begin{proof}[Proof of Lemma \ref{lem:one-dim}]

To prove the first part we are going to use Taylor expansion of the exponential:
\begin{eqnarray}
1 - 2\Phi(\rho^{-1}) &=& \int_{-\rho^{-1}}^{\rho^{-1}} \frac{1}{\sqrt{2\pi}} \exp\{-\frac{t^{2}}{2}\}dt\\
&=& \frac{1}{\sqrt{2\pi}} \sum_{k=0}^{\infty}\frac{(-1)^{k}}{2^{k}k!} \int_{-\rho^{-1}}^{\rho^{-1}} t^{2k}dt\\
&=& \sqrt{\frac{2}{\pi}} \sum_{k=0}^{\infty}\frac{(-1)^{k}}{(2k+1)2^{k}k!} \rho^{-(2k+1)}
\end{eqnarray}
Similarly we have that$\sqrt{\frac{2}{\pi}}\rho \exp\{-\rho^{-2}/2\} = \sqrt{\frac{2}{\pi}} \sum_{k=0}^{\infty}\frac{(-1)^{k}}{2^{k}k!} \rho^{-(2k-1)}$.
Putting it all together we get:
\[
C(\rho) = \sqrt{\frac{2}{\pi}} \sum_{k=0}^{\infty}\frac{(-1)^{k}}{2^{k}k!(2k+2)(2k+1)}  \rho^{-(2k+1)}
\]
To prove the  second part of the lemma we use the following bounds for gaussian tails.

\begin{lemma}[\cite{Feller1}]\label{lem:feller}
Let $X\sim \mathcal{N}(0,1)$ and $\Phi(x):=\P[X\geq x]$ for all $x>0$, then:
\begin{equation}
\frac{1}{\sqrt{2\pi}}\frac{1}{x}\left\{1-\frac{1}{x^{2}}\right\}e^{-\frac{x^{2}}{2}}\leq \Phi(x)\leq \frac{1}{\sqrt{2\pi}}\frac{1}{x}\left\{1-\left(\frac{1}{x^{2}}-\frac{1}{x^{4}}\right)\right\}e^{-\frac{x^{2}}{2}}
\end{equation} 
\end{lemma}

Starting from \eqref{eq:DIIM} and using   Lemma \ref{lem:feller} for $x=c^{-1}$ we arrive at:
\begin{equation}
1 - \sqrt{\frac{2}{\pi}}c\left(1-c^{2}(1-3c^{2})e^{-c^{-2}/2}\right)\leq p_{1}(c)\leq 1- \sqrt{\frac{2}{\pi}}c\left(1-c^{2}e^{-c^{-2}/2} \right)
\end{equation}
For the upper bound, the assumption that $c\leq\min\{ \frac{1}{\sqrt{2\ln(1/\delta)}}, \delta\}$ implies that $1-c^{2}e^{-c^{-2}/2}\geq 1-\delta^{3}$ and consequently $p_{1}(c)\leq e^{-\sqrt{\frac{2}{\pi}}(1-\delta^{3})c}$. For the lower bound we  use the fact that the function $f(c)=c^{2}(1-3c^{2})e^{-c^{2}/2}$ is decreasing for all $c\in [0,1/2]$ to  obtain $p_{1}(c)\geq 1-\sqrt{\frac{2}{\pi}}c$. The lower bound then follows by the inequalities $1-x\geq e^{-(1+\frac{\epsilon}{2})x}$ for $x\in[0,\frac{\epsilon}{1+\epsilon}]$  and $\sqrt{\frac{2}{\pi}}c\leq \delta\leq \frac{2\delta}{1+2\delta}$ for $\delta\leq 1/2$.
\end{proof}
\section{Scale-free estimator based on Andoni-Indyk LSH}\label{sec:gaussian}

Our framework of scale-free esimators is a natural desideratum when trying to approximate the ideal function for importance sampling. In this section, we show how to use the ``Ball-Carving'' LSH introduced by Andoni and Indyk~\cite{andoni2006near} for Euclidean distance to get a \emph{scale-free} estimator for the Gaussian Kernel.
\subsection{Ball-Carving LSH}
 Andoni and Indyk~\cite{andoni2006near} introduced a family of hash functions $
\mathcal{H}_{t}(w)$ parametrized by an integer $t\geq 2$, and a width $w>0$ such that the evaluation cost is bounded by $U_{t}=d2^{O(t\log t)}\log n$ and space usage by $O(U_{t}+dn)$. Their scheme partitions the space by randomly projecting into $t$-dimensions and then carving out balls of radius $w$ centered at random points. We refer to the resulting probability measure as $\nu_{t}$. Using a similar analysis as the one in \cite{andoni2006near} we show the following bounds on the collision probability $p_{t}(c):=\P_{h\sim\nu_{t}}[h(x)=h(y)]$  of two points $x,y$ at distance $cw$.

\begin{lemma}[Pointwise Bounds]\label{lem:AI-bounds}
 The function $p_{t}(c)$ is non-increasing for all $c\geq 0,t\geq 1$. Furthermore,  for all $t\geq 12$ and $\frac{16}{t+7}\leq c^{2} \leq 1$
\begin{align}
 p_{t}(c) & \geq 
\left\{\frac{1}{4\sqrt{t}}\frac{1}{c} (1-2e^{-\frac{9t}{100}})e^{-\frac{t-1}{8}\frac{1}{2-c^{2}}c^{4}}\right\}  e^{-\frac{t-1}{8}c^{2}}\label{eq:AI-lower}\\
p_{t}(c) & \leq  \left\{\frac{3}{\sqrt{t}}\frac{1}{c}\left(1+\frac{\sqrt{t}c}{3}e^{-\frac{9}{64}t}\right) e^{\frac{t-1}{8^{2}}c^{4}}  \right\}e^{-\frac{t-1}{8}c^{2}}\label{eq:AI-upper}
\end{align}
\end{lemma}
Observe that for small $c=O(1/\mathrm{poly}(\log n))$ and large enough $t=\mathrm{poly}(\log n)$, the dominant term is $e^{-\frac{t-1}{8}c^{2}}$ and drops exponentially with the squared of the distance as desired. However, in order for the time $U_{t}$ to compute the hash function  on query to be $n^{o(1)}$, we must have $t=o(\log(n))$. This is the main bottleneck that complicates a bit the application of the Andoni-Indyk hashing scheme for our purposes.

\subsection{Scale-free estimator for Gaussian Kernel}
\begin{theorem}[Gaussian Kernel] \label{thm:AI}
For all  $\beta\in (0,1]$ there exists a $(\beta, e^{O(R^{\frac{4}{3}}\log\log n)})$-scale free HBE for the Gaussian kernel $e^{-\|x-y\|^{2}}$ that has complexity $T=e^{O(R^{\frac{4}{3}}\log\log n)}$.
\end{theorem}
\begin{proof}
Set $t:=\max\{\lceil R^{\frac{4}{3}}\rceil,12\}$, $D= \lceil\frac{8\sqrt{t}R^{2}}{t-1}\rceil,  w=\sqrt{\frac{t-1}{8\beta}D}\geq \sqrt[4]{t}R$ and consider the HBE resulting from using the family $\mathcal{H}_{t}^{D}(w)$. The HBE can be evaluated in time $D2^{O(t\log\log(n))}$ and takes space $D2^{O(t\log\log(n))}\cdot n$. Next, we will show that for this selection of parameters we get a scale-free estimator for the Gaussian kernel. Consider two points at distance $r=c\cdot w$. First, we see that $c=\frac{r}{w}\leq \frac{R}{w}\leq \frac{1}{\sqrt[4]{t}}\leq 1$. For distances $r \geq \frac{4}{\sqrt[4]{t}}R \Rightarrow c^{2} \geq \frac{16}{t}$, our selection of $D$  results in the dominant term of the bounds given by Lemma  \ref{lem:AI-bounds} being $e^{-\frac{t-1}{8w^{2}}D\cdot r^{2}}=e^{-\beta r^{2}}$. Moreover,
\begin{equation}
e^{-\Theta(D\log(t))} e^{-\beta r^{2}} \leq p_{t}^{D}(c) \leq e^{\Theta(D\log(t))}  e^{-\beta r^{2}}
\end{equation}
Thus, we see that for the range $4R/\sqrt[4]{t}\leq r\leq R$ we have a scale-free estimator with $M=e^{O(D\log(t))}$. To get a bound for $ 0\leq r\leq R/\sqrt[4]{t}$ we use monotonicity of $p_{t}(c)$ to obtain:
\begin{align}
 M^{-1} \cdot e^{-16\beta \frac{R^{2}}{\sqrt{t}}} \leq  p_{t}^{D}(c) \leq  e^{\beta r^{2}} \cdot e^{-\beta r^{2}}
\end{align}
Since $D= \Theta(R^{2}/\sqrt{t})$ we see that we have constructed an $(\beta, e^{O(\log\log (n) R^{2}/\sqrt{t})})$-scale free estimator. Our selection of $t=R^{4/3}$ balances the complexity of evaluating the hashing function with the deviation from the ideal collision probabilities.
\end{proof}
The above theorem shows that as long as $R=O(\log^{\gamma}(n))$ with $0<\gamma<\frac{3}{4}$ the estimator is $(\beta, n^{o(1)})$-scale free. The regime of most interest is when $\gamma=1/2$, where polynomially small values of $\mu = n^{-\Omega(1)}$ are permissible. In this case, our estimator is $(\beta, e^{O(\log^{\frac{2}{3}}(n)\log\log n)})$-scale free.

\subsection{Pointwise bounds for Andoni-Indyk LSH}
Let $I(r , w)$ be ratio between the cap volume at relative distance $\alpha:=r/w $ from the origin of a $d$-dimensional sphere relative of radius $w$ and the volume of the sphere.
\begin{lemma}[Corollary 3.2~\cite{boroczky2003covering}]\label{lem:sphere}
Let $\alpha:=\frac{r}{w} \in [0,1]$, 
\begin{enumerate}
\item[(i)] for $\alpha \geq \frac{1}{\sqrt{d+1}}$ we have
\begin{equation}
\frac{\sqrt{d}}{3\sqrt{d+1}}\frac{1}{\sqrt{2\pi d}} \frac{1}{\alpha} (1-\alpha^{2})^{d/2} \leq I(r,w) \leq \frac{1}{\sqrt{2\pi d}} \frac{1}{\alpha } (1-\alpha^{2})^{d/2}
\end{equation}
\item[(ii)] for $\alpha \leq \frac{1}{\sqrt{d+1}}$,  $\frac{1}{2e\sqrt{2\pi}} \leq I (r, w) \leq \frac{1}{2}$
\end{enumerate}
\end{lemma}
We also require the following lemma on concentration properties of $\chi^{2}$-random variables.	
\begin{lemma}\label{lem:chi}
Let $Z_{1},\ldots, Z_{t}$ be independent $\chi^{2}$-random variables, then:
\begin{enumerate}
\item[(i)] $\P[\sum_{i=1}^{t}Z_{i} \leq \alpha \cdot t] \leq e^{-(\alpha+\ln(1/\alpha)-1) \frac{t}{2}}$ for all $\alpha <1$.
\item[(ii)] $\P[\sum_{i=1}^{t}Z_{i} > \alpha \cdot t] \leq e^{-(a-\ln(a)-1) \frac{t}{2}}$ for all $\alpha >1$.
\end{enumerate}
\end{lemma}
\begin{proof}[Proof of Lemma \ref{lem:AI-bounds}] To bound the collision probability we follow the analysis of Andoni and Indyk~\cite{andoni2006near}.
Given two points $p,q\in \R^{d}$ that are within distance $c\cdot w$ find the probability that they hash to the same hash bucket.
\begin{eqnarray}
p_{t}(c) &=& \int_{0}^{\infty} \P\left[h(p)=h(q)\left|\| p^{'} - q^{'}\| = \sqrt{\frac{x}{t}}c w\right.\right] P_{\chi_{t}^{2}}(x)dx\\
&=& \int_{0}^{\frac{4t}{c^{2}}}  P_{\chi_{t}^{2}}(x) \frac{I\left(\frac{1}{2}\sqrt{\frac{x}{t}}cw, w\right)}{1-I\left(\frac{1}{2}\sqrt{\frac{x}{t}}cw, w\right)}dx \\
&=& \int_{0}^{\frac{4}{c^{2}(t+1)}t} P_{\chi_{t}^{2}}(x) \frac{I\left(\frac{1}{2}\sqrt{\frac{x}{t}}cw, w\right)}{1-I\left(\frac{1}{2}\sqrt{\frac{x}{t}}cw, w\right)}dx + \int_{\frac{4}{c^{2}(t+1)}t}^{\frac{4}{c^{2}}t}P_{\chi_{t}^{2}}(x) \frac{I\left(\frac{1}{2}\sqrt{\frac{x}{t}}cw, w\right)}{1-I\left(\frac{1}{2}\sqrt{\frac{x}{t}}cw, w\right)}dx
\end{eqnarray}
where $P_{\chi_{t}^{2}}(x) = \frac{x^{t/2-1}}{\Gamma(t/2)2^{t/2}}e^{-x/2}$ is the density function of a $\chi^{2}$-random variable with $t$-degrees of freedom.
We obtain upper and lower bounds for each term separately.  Let $A$ and $B$ denote the first and second term respectively.  Using Lemma \ref{lem:sphere} and Lemma \ref{lem:chi} for $\alpha =\frac{4}{c^{2}}\frac{1}{t+1}$, we get   the following bounds for the first term:
\begin{equation}
0 \leq A \leq \left\{\begin{matrix}
e^{-\frac{1}{4}t} & \textrm{for} \ c^{2}\geq \frac{16}{t+1}\\
1 & \textrm{otherwise}
\end{matrix}\right.
\end{equation}
\subsubsection{Upper bound}
Using the fact that $I(r,w)\leq \frac{1}{2}$ for all $0\leq r \leq w$ and Lemma \ref{lem:sphere} we get
\begin{eqnarray}
B &\leq& 2 \int_{\frac{4}{c^{2}(t+1)}t}^{\frac{4t}{c^{2}}}P_{\chi^{2}_{t}}(x)\cdot I\left(\frac{1}{2}\sqrt{\frac{x}{t}}cw, w\right)dx\\
&\leq &2\int_{\frac{4}{c^{2}(t+1)}t}^{\frac{4t}{c^{2}}}P_{\chi^{2}_{t}}(x)\cdot \frac{1}{\sqrt{2\pi t}} \frac{1}{\frac{1}{2}\sqrt{\frac{x}{t}}c} \left(1-\left(\frac{1}{2}\sqrt{\frac{x}{t}}c\right)^{2}\right)^{t/2}dx\\
&\leq &\frac{4}{\sqrt{2\pi}c}\int_{\frac{4}{c^{2}(t+1)}t}^{\frac{4t}{c^{2}}}P_{\chi^{2}_{t}}(x)\cdot  \frac{1}{\sqrt{x}} \exp\left(-\frac{1}{2}x\frac{c^{2}}{4} \right)dx\\
&=& \frac{4}{\sqrt{2\pi}c}\int_{\frac{4}{c^{2}(t+1)}t}^{\frac{4t}{c^{2}}}\frac{x^{t/2-1}}{\Gamma(t/2)2^{t/2}}e^{-\frac{x}{2}}\cdot  \frac{1}{\sqrt{x}} \exp\left(-\frac{1}{2}x\frac{c^{2}}{4} \right)dx\\
&\leq & \frac{2\Gamma((t-1)/2)}{\sqrt{\pi}c\Gamma(t/2)}\frac{1}{(1+\frac{c^{2}}{4})^{(t-1)/2}}\int_{0}^{\frac{4t}{c^{2}}(c^{2}/4+1)}P_{\chi^{2}_{t-1}}(x)\\
&\leq & \frac{2\Gamma((t-1)/2)}{\sqrt{\pi}c\Gamma(t/2)}\frac{1}{(1+\frac{c^{2}}{4})^{(t-1)/2}}\\
&\leq & \frac{8\cdot 2}{7\sqrt{\pi \frac{t-1}{2}}c}\frac{1}{(1+\frac{c^{2}}{4})^{(t-1)/2}}\label{eq:up_bound_1}
\end{eqnarray}
where we used the fact that $\frac{8}{7}\frac{1}{\sqrt{\alpha}}\geq \frac{\Gamma(\alpha)}{\Gamma\left(\alpha +\frac{1}{2}\right)} \geq \frac{1}{\sqrt{\alpha}}$ for $\alpha\geq 1$. Further, due to the inequalities $1-x\leq \frac{1}{1+x}\leq e^{-\frac{2x}{2+x}}$ for $x\geq 0$ we get:
\begin{align}
\frac{1}{(1+\frac{c^{2}}{4})^{\frac{t-1}{2}}} & \leq e^{-\frac{t-1}{2}\frac{2}{2+\frac{c^{2}}{4}}\frac{c^{2}}{4}}=e^{-\frac{t-1}{8}c^{2}\frac{1}{1+c^{2}/8}}\leq e^{-\frac{t-1}{8}(1-\frac{c^{2}}{8})c^{2}}\leq e^{+\frac{t-1}{64}c^{4}}\cdot e^{-\frac{t-1}{8}c^{2}}
\end{align}
Substituting the last bound back to \eqref{eq:up_bound_1} and using   $\frac{16\sqrt{2}}{7\sqrt{\pi}}\frac{1}{\sqrt{t-1}}\leq \frac{3}{\sqrt{t}}$ for $t\geq 2$, results in
\begin{align}
p_{t}(c)&\leq e^{-\frac{t}{4}}+\frac{3}{\sqrt{t}c}e^{\frac{t-1}{64}c^{4}} \cdot e^{-\frac{t-1}{8}c^{2}}\\
&\leq \left(1+\frac{\sqrt{t}c}{3}e^{-\frac{t-1}{64}c^{4}+\frac{t-1}{8}c^{2}-\frac{t}{4}}\right)\frac{3}{\sqrt{t}c}e^{\frac{t-1}{64}c^{4}} \cdot e^{-\frac{t-1}{8}c^{2}}
\end{align}
Noting that for $c\leq 1$ the exponent in the parenthesis is at most $e^{-9t/64}$ results in the desired bound. 
\subsubsection{Lower Bound}
Using the fact that $1-I(r,w)\leq 1$ for all $0\leq r\leq w$ and Lemma \ref{lem:sphere} we get:
\begin{eqnarray}
B &\geq&  \int_{\frac{4}{c^{2}(t+1)}t}^{\frac{4t}{c^{2}}}P_{\chi^{2}_{t}}(x)\cdot I\left(\frac{1}{2}\sqrt{\frac{x}{t}}c, t\right)dx\\
&\geq&\frac{\sqrt{t}}{3\sqrt{t+1}}
\int_{\frac{4}{c^{2}(t+1)}t}^{\frac{4t}{c^{2}}}P_{\chi^{2}_{t}}(x)\cdot \frac{1}{\sqrt{2\pi t}} \frac{1}{\frac{1}{2}\sqrt{\frac{x}{t}}c} \left(1-\left(\frac{1}{2}\sqrt{\frac{x}{t}}c\right)^{2}\right)^{t/2}dx\\
&= &\frac{\sqrt{t}}{3\sqrt{t+1}}
\frac{2}{\sqrt{2\pi} c}\int_{\frac{4}{c^{2}(t+1)}t}^{\frac{4t}{c^{2}}}P_{\chi^{2}_{t}}(x)\cdot  \frac{1}{\sqrt{x}} \left(1-\frac{x}{4t}c^{2}\right)^{t/2}dx\\
&=& \frac{\sqrt{t}}{3\sqrt{t+1}} \frac{2}{\sqrt{2\pi} c}
\int_{\frac{4}{c^{2}(t+1)}t}^{\frac{4t}{c^{2}}}P_{\chi^{2}_{t}}(x)\cdot \frac{1}{\sqrt{x}} \left(\frac{1}{1+\frac{\frac{x}{4t}c^{2}}{1-\frac{x}{4t}c^{2}}}\right)^{t/2}dx\\
&\geq & \frac{\sqrt{t}}{3\sqrt{t+1}}\frac{2}{\sqrt{2\pi} c}
\int_{\frac{4}{c^{2}(t+1)}t}^{\frac{4t}{c^{2}}}P_{\chi^{2}_{t}}(x)\cdot \frac{1}{\sqrt{x}} \exp\left(-\frac{t}{2}\frac{\frac{x}{4t}c^{2}}{1-\frac{x}{4t}c^{2}}\right)dx\\
&\geq &\frac{\sqrt{t}}{3\sqrt{t+1}}\frac{2}{\sqrt{2\pi} c}
\int_{\frac{4}{c^{2}(t+1)}t}^{\frac{R-1}{R}\frac{4t}{c^{2}}}P_{\chi^{2}_{t}}(x)\cdot \frac{1}{\sqrt{x}} \exp\left(-\frac{1}{2}x\frac{c^{2}}{4}R\right)dx\\
&= &\frac{\sqrt{t}}{3\sqrt{t+1}}\frac{\sqrt{2}\Gamma((t-1)/2)}{\sqrt{\pi}c\Gamma(t/2)}\frac{1}{\left(1+\frac{c^{2}}{4}R\right)^{(t-1)/2}} \int_{\frac{1}{t+1}(R+\frac{4}{c^{2}})t}^{(R-1)(1+\frac{4}{Rc^{2}})t}P_{\chi^{2}_{t-1}}(x)dx
\end{eqnarray}
where in the penultimate inequality we used $1-x\leq 1+ Rx$ for $x\leq (R-1)/R$ and $R>1$. To bound the integral we write:
\begin{eqnarray}
\int_{\frac{1}{t+1}(R+\frac{4}{c^{2}})t}^{(R-1)(1+\frac{4}{Rc^{2}})t}P_{\chi^{2}_{t-1}}(x)dx &=& 1 - \int_{0}^{\frac{1}{t+1}(R+\frac{4}{c^{2}})t}P_{\chi^{2}_{t-1}}(x)dx - \int_{\frac{R-1}{R}\frac{4t}{c^{2}}}^{\infty}P_{\chi^{2}_{t-1}}(x)dx
\end{eqnarray}
We set $R=1+\delta(c)=1+\frac{2}{\frac{4}{c^{2}}-2}$ and require that:
\begin{align}
(R-1)\frac{4}{Rc^{2}}\geq 2 &\Rightarrow \delta(c) \geq  \frac{2}{\frac{4}{c^{2}}-2}\\
\frac{1}{t+1}(R+\frac{4}{c^{2}})\leq \frac{1}{2} &  \Rightarrow  \frac{16}{t+3+\sqrt{t^{2}-10t-7}} \leq c^{2} \leq \frac{16}{t+3-\sqrt{t^{2}-10t-7}}
\end{align}
For $t\geq 12$ the inequalities are satisfied for all $ \frac{16}{t+7}\leq c \leq 1$. For such $c$, we apply Lemma \ref{lem:chi} for $\alpha_{>}=2$ and $\alpha_{<}=1/2$ and obtain the bound:
\begin{align}
\int_{\frac{1}{t+1}(R+\frac{4}{c^{2}})t}^{(R-1)(1+\frac{4}{Rc^{2}})t}P_{\chi^{2}_{t-1}}(x)dx \geq 1-2e^{-\frac{9t}{100}}
\end{align}
Next, we lower bound the main term $(1+c^{2}R/4)^{-(t-1)/2}$ using $(1+x)^{-1}\geq e^{-x}$:
\begin{align}
\frac{1}{(1+c^{2}R/4)^{-(t-1)/2}}\geq e^{-\frac{t-1}{2}\frac{c^{2}}{4}(1+\delta(c))} = e^{\frac{t-1}{8}(\frac{2}{c^{2}}-1)^{-1}c^{2}}\cdot e^{-\frac{t-1}{8}c^{2}}
\end{align}
To complete the proof, we note that $\frac{\sqrt{2t}}{3\sqrt{(t+1)\pi}}\frac{\Gamma((t-1)/2)}{\Gamma(t/2)}\geq \frac{1}{4\sqrt{t}}$.
\end{proof}
\subsubsection{Concentration of $\chi^{2}$-random variables}
\begin{proof}[Proof of Lemma \ref{lem:chi}]
The property we   use is that for all $\lambda<1/2$, $\E[e^{\lambda Z_{i}}]=(1-2\lambda)^{-1/2}$.
Let $\lambda_{-}= \frac{t}{2}(\alpha-1)>0$ for $\alpha>1$:
\begin{align}
\P[\sum_{i=1}^{t}Z_{i}\leq \alpha t ] & = \P[\frac{1}{t}\sum_{i=1}^{t}(\alpha-Z_{i})\geq 0]\\
&=\P\left[\exp\left(\left\{\sum_{i=1}^{t}\frac{\lambda_{-}}{t}(\alpha-Z_{i})\right)\right\} \geq 1 \right]\\
&\leq \E[\exp(\frac{\lambda_{-}}{t}(\alpha-Z_{1}))]^{t}\\
&=e^{\lambda_{-}\alpha-\frac{t}{2}\log(1+2\frac{\lambda_{-}}{t})}\\
& = e^{-\left(\alpha+\log(\frac{1}{\alpha})-1 \right)\frac{t}{2}}
\end{align}
Similarly, for $\alpha>1$ and $\lambda_{+}=\frac{t}{2}(1-\frac{1}{\alpha})<t/2$:
\begin{align}
\P[\sum_{i=1}^{t}Z_{i} \geq \alpha t] & = \P[\frac{1}{t}\sum_{i=1}^{t}(Z_{i}-\alpha)\geq 0]\\
&=\P\left[\exp\left(\left\{\sum_{i=1}^{t}\frac{\lambda_{+}}{t}(Z_{i}-\alpha)\right)\right\} \geq 1 \right]\\
&\leq \E[\exp(\frac{\lambda_{+}}{t}(Z_{1}-\alpha))]^{d'}\\
&=e^{-\alpha^{2}\lambda_{+}-\frac{t}{2}\log(1-2\frac{\lambda_{+}}{t})}\\
&=e^{-\left(\alpha-\log(\alpha)-1\right)\frac{t}{2}}
\end{align}
\end{proof}
\section{Fast Kernel-Matrix Vector Multiplication}\label{sec:matrix}

Given a kernel function $k:\R^{d}\times \R^{d}$ and a set of points $P=\{x_{1},\ldots, x_{n}\}$, let $K=\{K(x_{i},x_{j})\}_{i,j\leq n}$ denote the matrix with the pairwise evaluations of a kernel. Given a vector $z\in \R^{n}$, the problem of approximate Kernel Matrix-Vector Multiplication (aKMVM) is to obtain an approximation to $y = Kz$. Due to linearity, we can always rescale the vectors without changing the problem, thus we may assume without loss of generality that $\|z\|_{1}=1$. This problem is important as many machine learning applications involve the mutltiplication of a vector with a dense Kernel matrix. Very often this operation is the computational bottleneck.

Observe that if $z=\frac{1}{n}\1$ then the problem is equivalent to estimating the kernel density of all the points in the dataset $P$, a problem that we can solve relatively fast. We show how one can adapt the techniques from this paper to provide a solution to the aKMVM problem. 
\begin{theorem}\label{thm:matrix}
Given a $V$-bounded  HBE for a kernel $k$ with complexity $T$, there exists an algorithm that given  a dataset $P$  and a \emph{non-negative} weight vector $z$ such that $\|z\|_{1} = 1$  can compute a vector $\hat{y}$ in time $\tilde{O}(\frac{1}{\epsilon^{2}}V(\epsilon \tau)\cdot nT)$ using  space $\tilde{O}(\frac{1}{\epsilon^{2}}V(\epsilon \tau)\cdot nT)$  such that with probability at least $1-n^{-1}$ for all $i\in [n]$ it holds $|\hat{y}_{i}-y_{i}|\leq 3\epsilon\tau +\epsilon |y_{i}|$ and 
\begin{equation}
  \|\hat{y}-y\|_{p}\leq \epsilon \left(3\tau n^{1/p}+\|y\|_{p}\right)
\end{equation}
\end{theorem}
\begin{remark} Let $z_{+}=(z)_{0}$ and $z_{-}=(-z)_{0}$ be the positive and negative ``parts" of vector $z$.
We can apply this algorithm separately to estimate $y_{+}=Kz_{+}$ and $y_{-}=Kz_{-}$ to obtain a vector $\hat{y}=\hat{y}_{+}-\hat{y}_{-}$ that satisfies: $\|\hat{y} - K z\|_{p}\leq \epsilon (6\tau n^{1/p} + \|y_{-}\|_{p}+\|y_{+}\|_{p})$. 
\end{remark}
\begin{proof} 
We consider intervals $I_{\ell}=[2^{-\ell},2^{-\ell+1}]$ for $\ell=1,\ldots, L=\log_{2}(n/\tau^{'})$, where   $\tau^{'} = \epsilon\tau$, and we set $I_{0} = [0, \frac{\tau^{'}}{n})$. We  partition points according to which interval their corresponding ``weights'' $z_{i}$ belong to. Let $S_{0},S_{1},\ldots, S_{L}$ be the corresponding sets. For any point $i$, we have 
\begin{align}
y_{i}^{*}&=  \sum_{\ell=0}^{L} \sum_{j\in S_{\ell}} K(x_{i},x_{j}) z_{j}\nonumber\\
& = \sum_{\ell=1}^{L}Z_{\ell} \left(\sum_{j\in S_{\ell}}k(x_{i},x_{j})\frac{z_{j}}{\sum_{s\in S_{\ell}}z_{s}}\right) + \delta_{i}
\end{align}
where $\delta_{i} = \sum_{j \in S_{0}} K(x_{i},x_{j})z_{j} \leq \tau^{'}$ and  $Z_{\ell}=\sum_{s\in S_{\ell}}z_{s}$. By construction it holds that  $|S_{\ell}|2^{-\ell}\leq Z_{\ell}\leq |S_{\ell}|2^{-\ell+1}$  and  within each set $S_{\ell}$ the weights differ by at most a factor of two. The numbers $Z_{\ell}$ can be computed with a linear pass on the data. Thus, our problem can be expressed as a weighted version of $L$ KDE problems, where the $\ell$-th problem asks to compute an approximation to :
\begin{equation}
\mathrm{KDE}_{S_{\ell}}^{z}(x)=\frac{1}{Z_{\ell}}\sum_{y\in S_{\ell}} k(x,y) z_{y}
\end{equation}

Given a hashing scheme $\mathcal{H}$ with collision probabilities $p_{i}$ we evaluate the hash function during the preprocessing step only on $S_{\ell}$ and then for a query point $x\in P$ we define the following estimator. 
\begin{equation}
Z_{h} = \frac{z_{i}}{Z_{\ell}}\frac{k(x,y_{I})}{p_{I}}|H(x)|
\end{equation}
where as before I is a random index from  $H(x)\subseteq S_{\ell}$. The first two moments of the estimator are given by:
\begin{align}
\E[Z] &= \frac{1}{Z_{\ell}}\sum_{i\in S_{\ell}} k(x,y_{i})z_{i} \\
 \E[Z^{2}] &= \frac{1}{Z_{\ell}^{2}}\sum_{i\in S_{\ell}} \frac{k(x,y_{i})^{2}}{p_{i}}z_{i}^{2}\E[|H(x)||i\in H(x)]
\end{align}
Using the properties of the set $S_{\ell}$ we obtain the following upper bound on the variance of this estimator
\begin{equation*}
\E[Z^{2}] \leq \frac{2^{2\ell}}{|S_{\ell}|^{2}}2^{-2\ell+2} \sum_{i\in S_{\ell}} \frac{k(x,y_{i})^{2}}{p_{i}} \E[|H(x)||i\in H(x)]
\end{equation*}
Now observe that the variance is at most $4$ times larger than if we would be trying to estimate $\mathrm{KDE}_{S_{\ell}}(x)$. Hence, given any $V$-bounded HBE for the kernel $k$ and set $S_{\ell}$ we can use the above modification to get a $4V$-bounded HBE for $\mathrm{KDE}_{S_{\ell}}^{z}(x)$. Invoking Theorem \ref{thm:main} with parameters $(\epsilon,\tau^{'}, \chi/(nL))$ we can get a data structure that can estimate   $\mathrm{KDE}^{z}_{S_{\ell}}(x_{i})$ with probability at least $1-\chi$ either within multiplicative accuracy $\epsilon$ (when AMR* has non-zero output) or with absolute accuracy $\tau^{'}$ (when AMR* outputs $0$) for all $i\in[n]$ (by union bound). Let $\mathcal{L}$ denote the set of indices of $[L]$ such that $Z_{\ell} \geq \frac{\tau^{'}}{|L|}$. For all $\ell \in \mathcal{L}$ we instantiate the data structure given by Theorem \ref{thm:main} for the set $S_{\ell}$ and use it to query all points in $P$. For each query point thus we get estimates $z_{\ell}(x_{i})$ for $\ell \in \mathcal{L}$ and set $z_{\ell}(x_{i})=0$ for $\ell\notin \mathcal{L}$. The overhead per-query of the whole process is at most a multiplicative factor $L=O(\log(n/\tau^{'}))$ compared to the case that we were creating a single data structure for the same problem. For any query $x_{i}$, we aggregate the estimates $z_{\ell}(x_{i})$ in the following manner $Z(x_{i}) = \sum_{\ell\in \mathcal{L}} Z_{\ell} \cdot z_{\ell}(x_{i})$. For all $z_{\ell}(x_{i})$ with probability at least $1-\chi$ it holds
\[
|z_{\ell}(x_{i}) - \mathrm{KDE}_{S_\ell}^{z}(x_{i})|\leq \max\{\tau^{'},\epsilon\cdot \mathrm{KDE}_{S_\ell}^{z}(x_{i})\}
\]

This implies the following bounds:
\begin{align*}
|Z(x_{i})-y_{i}^{*}| &\leq \sum_{\ell\in [L]} Z_{\ell}\cdot  |z_{\ell}(x_{i})-\mathrm{KDE}_{S_{\ell}}^{z}(x_{i})| +\delta_{i}\\
&\leq \sum_{\ell \notin \mathcal{L}} Z_{\ell} + \epsilon \sum_{\ell \in \mathcal{L}}Z_{\ell}(\tau+\mathrm{KDE}_{S_{\ell}}(x_{i}))+\delta_{i}\\
&\leq 3\epsilon \tau + \epsilon |y_{i}^{*}| 
\end{align*}
Summing over all indices and using triangle inequality gives $\|\hat{y}-y^{*}\|_{p} \leq \epsilon (3\tau n^{1/p} + \|y^{*}\|_{p})$. 
\end{proof}

\section{Lower Bound for Kernel Density Estimation}\label{sec:lower}

The nature of the KDE problem is quite analytical as it involves a simple summation of $n$ terms that depend smoothly on the distances of the query from the dataset. The Gaussian kernel, though a smooth function of the distances, is rapidly decreasing and  can  be thought of being an approximation to the indicator function $\I\left\{\|x-y\|^{2}\leq \sigma^{2}\right\}$. In particular, for two distances $r_{1}=\sigma, r_{2}>\sqrt{C} \sigma$ the kernel value varies from $e^{-1}=\Omega(1)$ to $e^{-C}=o(1)$ for any $C=\omega(1)$. This is the basic observation motivating reducing the \emph{Approximate Nearest Neighbor Search} problem to the KDE problem. 

\begin{definition}[$(r,c)$-ANNS] Given a metric space $\mathcal{M}$ and parameters $c>1$ and $r>0$, and a dataset of $n$ points $x_{1},\ldots, x_{n}$, our goal is to distinguish between the case that $d(x_{i},y)\leq r$ for some $i\in[n]$ and the case where for all $i\in[n], d(x_{i},y)\geq cr$. The query algorithm must output $1$ in the former case, $0$ in the latter case and may report anything if neither of the two cases hold.
\end{definition}

The complexity of ANNS has been a topic of ongoing research over the past two decades~\cite{panigrahy2010lower,andoni2017optimal}. The most popular model of computation to prove lower bounds for is the cell-probe model. In this model, we are allowed an arbitrary amount of preprocessing but only allowed to keep $m$ cells each with $w$ bits of information. The query algorithm then queries (adaptively) $t$ cells  and is required to produce the output. We will refer to this model as the \emph{$(m,w,t)$-cell probe model}. For ANNS, lower bounds in this model impose constraints on $m,w,t$ depending on $n,c$. The most general result in this area is given by the following theorem proved by Panigrahy, Talwar and Wieder~\cite{panigrahy2010lower}, whose estimates where improved by Andoni et al.~\cite{andoni2017optimal}.

\begin{theorem}[\cite{panigrahy2010lower,andoni2017optimal}]\label{thm:cell-probe}
 There exists a distribution over $(r,c)$-ANNS instances and $\gamma\in[0,1]$ such that any randomized algorithm in the $(m,w,t)$-cell probe model which is correct with probability at least a half on these instances, satisfies:
\begin{align*}
\frac{m^{t}w}{n} &\geq  \sup_{(q-1)(p-1)=(1-\frac{1}{c})^{2}, p,q\geq1}\left\{(\frac{\gamma}{t})^{q}m^{t(1+\frac{q}{p}-q)}\right\}
\end{align*}
\end{theorem}
The quality of those bounds deteriorate as $t$ the number of probes increases. The bound for $t=1$ is optimal and recently Andoni et al.\cite{andoni2017optimal} gave an optimal bound up to sub-polynomial factors for $t=2$ in the regime of $w=O(\log n)$. The distribution for which the lower bounds are proved is the following.

\begin{definition}[(n,d,c)-random instance]Let $P$ be $n$ random points from the boolean hypercube $\{0,1\}^{d}$ with $d=\omega(\log n)$. The query point $x$ is generated by picking a random point $y\in P$ and then generating a $\rho$-correlated point $y_{\rho}$ by keeping each bit of $y$ independently with probability $\rho = 1-\frac{1}{c}$.
\end{definition}

Our strategy of providing lower bounds for KDE is to show that for the specific distribution over instances used by Panigrahy et al. one can use an algorithm for KDE that would solve the ANNS problem with  more than $1/2$ probability.

\begin{theorem}[ANNS to KDE]\label{thm:reduction}
For $\mu \in (0, \frac{1}{4}]$, $\epsilon\in [\frac{\mu}{1+2\mu},\frac{1}{5}]$ and $\delta \in (\frac{1}{2}+\mu,1]$, any algorithm that solves the $(\mu,\epsilon, \delta)$-KDE problem  solves also the $\left(n,d,c\right)$-random instance of ANNS with $n=\lceil\frac{1}{\mu}\rceil$, $d=\Theta(\log^{3}(n))$,  and $c=\Theta(\frac{\log(n)}{\log(1/\epsilon)})$ with probability at least $1/2$.
\end{theorem}
\begin{proof} We define the following parameters:  $\epsilon_{1}=\epsilon$, $\epsilon_{2}=2 \epsilon$, $r= \log(\frac{n}{\sqrt{(1-\epsilon_{1})(1-\epsilon_{2})}})$, $d= \lceil\frac{18 \log(n)}{\ln^{2}\left(\sqrt{\frac{1-\epsilon_{1}}{1-\epsilon_{2}}}\right)} r^{2}\rceil$, $\sigma^{2} = \frac{1}{r}\frac{d}{2}$,  $u=\sqrt{\frac{\log^{2}(\frac{1-\epsilon_{1}}{1-\epsilon_{2}})}{4\log(\frac{1}{4\epsilon})}}\leq \frac{1}{3}$ and $c = \frac{(1+u)}{\log(1/4\epsilon )}2r$. Using these parameters we bound the values of the kernel for different pair of points. Let $P$ be a collection of $n=\lceil \frac{1}{\mu}\rceil$ random points on the $d$-dimensional hypercube and $y_{\rho}$ denotes as before a $\rho$-correlated point. For simplicity we assume that $1/\mu$ is an integer.
\begin{claim}
For all $y\neq z\in P$, $(1-\epsilon_{2})\mu \leq k_{\sigma}(y_{\rho},z) \leq (1-\epsilon_{1})\mu$ and $k_{\sigma}(y_{\rho},y)\geq 4\epsilon$ with probability at least $1- n^{-1}$ over the randomness in $P$ and $y_{\rho}$.
\end{claim}
\begin{proof} For all $z\in P\setminus\{y\}$
let $X_{i}=\mathbb{I}[(y_{\rho})_{i}\neq z_{i}]$ with $\E[X_{i}]= \frac{1}{2}$ and define $X = \sum_{i=1}^{d}X_{i}$. In this case,  $k_{\sigma}(y_{\rho},z)= \exp(- \|z-y_{\rho}\|^{2}/\sigma^{2}) = \exp(-\sum_{i=1}^{d}\mathbb{I}[(y_{\rho})_{i}\neq z_{i}]/\sigma^{2})=\exp(-X/\sigma^{2})$.
\begin{align}
\P[k_{\sigma}(y_{\rho},z)> (1-\epsilon_{1})\mu]= \P\left[X<\left\{\frac{2\sigma^{2}}{d} \log(\frac{1}{(1-\epsilon_{1})\mu})\right\}\cdot \frac{d}{2}\right]
\end{align}
where $ \frac{2\sigma^{2}}{d} \log(\frac{1}{(1-\epsilon_{1})\mu}) = \frac{\log n + \log(\frac{1}{1-\epsilon_{1}})}{\log n + \frac{1}{2}\log(\frac{1}{1-\epsilon_{1}})+\frac{1}{2}\log(\frac{1}{1-\epsilon_{2}})} =  1- \frac{\log(\frac{1-\epsilon_{1}}{1-\epsilon_{2}})}{2\log(\frac{n}{\sqrt{(1-\epsilon_{1})(1-\epsilon_{2})})})}$. By Chernoff bounds we get:
\begin{align}
\P[k_{\sigma}(y_{\rho},z) > (1-\epsilon_{1})\mu]\leq \exp\left\{-\frac{\log^{2}(\frac{1-\epsilon_{1}}{1-\epsilon_{2}})}{3\cdot 4\log^{2}(\frac{n}{\sqrt{(1-\epsilon_{1})(1-\epsilon_{2})})})}\frac{d}{2}\right\}\leq n^{-3}
\end{align}
Similarly, 	 we have that $\P[k_{\sigma}(y_{\rho},z)> (1-\epsilon_{1})\mu]= \P\left[X>\left\{\frac{2\sigma^{2}}{d} \log(\frac{1}{(1-\epsilon_{2})\mu})\right\}\cdot \frac{d}{2}\right]$ 
where $\frac{2\sigma^{2}}{d} \log(\frac{1}{(1-\epsilon_{2})\mu})=1+ \frac{\log(\frac{1-\epsilon_{1}}{1-\epsilon_{2}})}{2\log(\frac{n}{\sqrt{(1-\epsilon_{1})(1-\epsilon_{2})})})}$. By Chernoff bounds we get:
\begin{align}
\P[k_{\sigma}(y_{\rho},z) > (1-\epsilon_{2})\mu]\leq \exp\left\{-\frac{\log^{2}(\frac{1-\epsilon_{1}}{1-\epsilon_{2}})}{3\cdot 4\log^{2}(\frac{n}{\sqrt{(1-\epsilon_{1})(1-\epsilon_{2})})})}\frac{d}{2}\right\}\leq n^{-3}
\end{align}
Next, let $Y_{i}=\mathbb{I}[(y_{\rho})_{i}\neq y_{i}]$ with $\E[Y_{i}]=\frac{1}{c}$ and define $Y=\sum_{i=1}^{d}Y_{i}$.  Then, 
\[
\P[k_{\sigma}(y_{\rho},y)<4\epsilon] = \P\left[Y> \left\{\frac{c\sigma^{2}}{d}\log(\frac{1}{4\epsilon})\right\}\frac{d}{c}\right]
\]
where $\frac{c\sigma^{2}}{d}\log(\frac{1}{4\epsilon})= \frac{1+u}{\log(1/4\epsilon)}\frac{\sigma^{2}}{d}\log(\frac{1}{4\epsilon})2r= 1+u$. By Chernoff bounds we get for all $n\geq  4$
\begin{align}
\P[k_{\sigma}(y_{\rho},z) < 4\epsilon]\leq \exp\left\{-\frac{u^{2}}{3}\frac{d}{c}\right\}\leq e^{-\frac{9}{4}\log^{2}(n)}\leq n^{-3}
\end{align}
Taking union bound for all $2\binom{n}{2}+n=n^{2}$ events completes the proof.
\end{proof}
Conditionally on these events, a no-instance of $c$-ANN problem, has kernel density at most $(1-\epsilon_{1}) \mu$, whereas a yes-instance has density at least $(1-\epsilon_{2})\frac{n-1}{n } \mu+\frac{1}{n}4\epsilon   = (1+2\epsilon(1+\frac{1}{n})-\frac{1}{n}))\mu\geq (1+\epsilon)\mu$ as long as $\epsilon(1+\frac{2}{n})\geq \frac{1}{n}$. We therefore, get that any data structure for the $(\mu,\epsilon,\delta)$-KDE problem can also solve the $c$-ANN problem for random instances with probability at least $1-\delta- n^{-1}$.
\end{proof}
\begin{corollary}For $\mu \in (0, \frac{1}{4}]$, $\epsilon\in [\frac{\mu}{1+2\mu},\frac{1}{5}]$ and $\delta \in (\frac{1}{2}+\mu,1]$, any algorithm that solves the $(\mu,\epsilon, \delta)$-KDE problem in the $(m,1,w)$-cell probe model must satisfy $
(m\cdot w) \geq \Omega(\frac{1}{\mu})$.
\end{corollary}
\begin{proof} 
Using the parameters defined in Theorem \ref{thm:reduction}, we invoke Theorem \ref{thm:cell-probe} with $p=\frac{1}{1-\rho^{2}}\geq 1$, $q=2-\rho^{2}\geq 1$ where $\rho^{2}=(1-\frac{1}{c})^{2}$ to obtain:
\begin{align}
m &\geq \left(\frac{1}{\mu w}\right)^{\frac{p}{p-1}\frac{1}{q}} \gamma^{\frac{p}{p-1}}\\
&=\left(\frac{1}{\mu w}\right)^{\frac{1}{\rho^{2}}\frac{1}{2-\rho^{2}}} \gamma^{\frac{1}{\rho^{2}}}\\
&=\left(\frac{1}{\mu w}\right)^{\frac{1}{1-\frac{2}{c}+\frac{1}{c^{2}}}\frac{1}{1+\frac{2}{c}-\frac{1}{c^{2}}}} \gamma^{\frac{1}{\rho^{2}}}\\
& =\left(\frac{1}{\mu w}\right)^{\frac{1}{1-\left(\frac{2}{c}-\frac{1}{c^{2}}\right)^{2}}} \gamma^{\frac{1}{\rho^{2}}}\\
&\geq \frac{1}{\mu w}\gamma^{\frac{1}{\rho^{2}}}
\end{align}
As $\gamma$ is a constant and $1/\rho^{2}=O(1)$ the  statement follows.
\end{proof}

\paragraph{Lower Bounds on Adaptive Core-sets} The above bound becomes more interesting if one considers the following class of estimators for KDE. Given $P$, we may perform any amount of preprocessing and only store $m$ sets $S_{1},\ldots,S_{m}$ (of arbitrary points) each of size at most $\frac{w}{d}$ . Given a query $x$, our estimator picks one of those sets $i(x)$ and produces an approximation to the Kernel Density by only a function of $(S_{i(x)},x)$, an example of such an estimate could also be something as simple as $\mathrm{KDE}_{S_{i(x)}}(x)$, or weighted versions such as $\sum_{y\in S_{i(x)}}
W(x,y)K(x,y)$. The above lower bound shows that either we must have many such sets (large space) or each set must be large itself (query time). Of particular interest is the case $m=1$ where there is a single set. In that case, we have a lower bound on the size of a core set. Hence, even for $t=1$ the lower bound gives non-trivial results for an interesting class of estimation algorithms.
\section{H\"{o}lder Inequalities}\label{sec:holder}
\begin{lemma}[H\"{o}lder's Inequality] Let $p,q>0$ such that $\frac{1}{p}+\frac{1}{q}=1$, then $|\sum_{i}x_{i}y_{i}|\leq \|x\|_{p}\|x\|_{q}$.
\end{lemma}
\begin{corollary}\label{col:holder} Let  $\beta \in [0,1]$ and $p\geq q>0$, then for all $x\in \R^{n}$
\begin{eqnarray} 
\|x\|_{\beta}^{\beta} & \leq & \|x\|_{1}^{\beta} \cdot n^{ 1-\beta }\\
\|x\|_{p}^{p} &\leq &\|x\|^{q}_{q}\cdot  \|x\|_{\infty}^{p-q}
\end{eqnarray}
\end{corollary}

\begin{proof} The first inequality is proved by fixing $\|x\|_{1}=a>0$, performing the change of variables $u_{j} = \left(\frac{x_{j}}{a}\right)^{\beta}$, so that $\|u\|_{\frac{1}{\beta}}^{\frac{1}{\beta}} = 1 $, and finally applying H\"{o}lder's inequality for $p=\frac{1}{\beta},q=\frac{1}{1-\beta}$. The second inequality follows by applying H\"{o}lder on $\sum_{i=1}^{n} |x_{i}|^{p-q}|x_{i}|^{q}$.
\end{proof}

\begin{lemma}[two-sided Holder]
 Let  $v,w\in \R^{n}$ be strictly positive vectors, then for any two sets $S,S^{'}\subseteq [n]$:
\begin{equation}
\sum_{i\in S, j\in S^{'}}A_{ij}x_{i}x_{j} \leq  
   \|x\|_{v,1}\|x\|_{w,1} \cdot \max_{i\in S,j\in S^{'}}\left\{\frac{|A_{ij}|}{v_{i}w_{j}}\right\}
\end{equation}
\end{lemma}
\begin{proof} We start by setting w.l.o.g. $\|x\|_{w,1}=\alpha$ and $\|x\|_{v,1}=\beta$ with $\alpha,\beta >0$.
\begin{align*}
\sup_{\tiny \begin{matrix}
\|x\|_{w,1}= \alpha\\ \|x\|_{v,1}= \beta
\end{matrix}}\left\{x_{S}^{\top}A_{S,S^{
'}}x_{S^{'}} \right\} = \sup_{\tiny\begin{matrix}
\|x\|_{w,1}= \alpha\\
\|y\|_{v,1}= \beta\\
y=x
\end{matrix}}\left\{x_{S}^{\top}Ay_{S^{'}} \right\} \leq \sup_{\tiny\begin{matrix}
\|x\|_{w,1}= \alpha\\
\|y\|_{v,1}= \beta\\
\end{matrix}}\left\{x_{S}^{\top}Ay_{S^{'}} \right\}
\end{align*}
We next peform the change of variables $\tilde{x} = \frac{w_{i}}{\alpha} x_{i}$ and $\tilde{y} = \frac{v_{i}}{\beta} y_{i}$ so that $\|\tilde{x}\|_{1}=\|\tilde{y}\|_{1}=1$. We get
\begin{align*}
\sup_{\tiny\begin{matrix}
\|x\|_{w,1}= \alpha\\
\|y\|_{v,1}= \beta\\
\end{matrix}}\left\{x_{S}^{\top}Ay_{S^{'}} \right\}&= \alpha \beta  \cdot \sup_{\tiny\begin{matrix}
 \|\tilde{x}\|_{1}=1\\
\|\tilde{y}\|_{1}=1 
\end{matrix}}\left\{\sum_{i\in S}\tilde{x}_{i}\left(\sum_{j\in S^{'}}\tilde{y}_{j}\frac{A_{ij}}{w_{i}v_{j}} \right) \right\}\\
& \leq \alpha \beta  \cdot \sup_{\tiny\begin{matrix}
 \|\tilde{x}\|_{1}=1\\
\|\tilde{y}\|_{1}=1 
\end{matrix}}\left\{\|\tilde{y}_{S^{'}}\|_{1}\sum_{i\in S}\tilde{x}_{i}\max_{j\in S^{'}}\left(\frac{|A_{ij}|}{w_{i}v_{j}}\right)  \right\}\\
& \leq \alpha \beta  \cdot \sup_{
 \|\tilde{x}\|_{1}=1}\left\{\|\tilde{x}_{S}\|_{1}\max_{i\in S,j\in S^{'}}\left(\frac{|A_{ij}|}{w_{i}v_{j}}\right)  \right\}\\
& = \alpha \beta  \cdot \max_{i\in S,j\in S'}\left\{\frac{|A_{ij}|}{w_{i}v_{j}} \right\}
\end{align*}
where we applied H\"{o}lder's inequality in the second and third inequality.
\end{proof}
\begin{lemma}[Monotone H\"{o}lder] For every integer $n\geq 1$,   $\beta\in[\frac{1}{2},1]$,    and  for all $x\in \R^{n}$ such that   $|x_{1}|\geq |x_{2}|\geq \ldots \geq |x_{n}|$  , we have
\begin{equation}
 \sum |x_{i}|^{\frac{2-\beta}{\beta}}\left(i+\sum_{j>i}\frac{|x_{j}|}{|x_{i}|}\right) \leq n^{\beta}\cdot  \left(\sum_{i=1}^{n}|x_{i}|^{\frac{1}{\beta}}\right)^{2-\beta}
\end{equation}
where equality holds for $x^{*} = c\1$ for any $c\neq 0$.
\end{lemma}
\begin{proof}  Without loss of generality we may assume that $x\in \R_{+}^{n}$ and that $\|x\|_{\frac{1}{\beta}}^{\frac{1}{\beta}}=a$ for $a> 0$. The claim is equivalent to proving the following inequality for all $n\geq 1$:
\begin{equation}
f^{\beta}_{n}(a) := \sup_{\begin{matrix}
\|x\|^{\frac{1}{\beta}}_{\frac{1}{\beta}}=a \\ 
x_{1}\geq \ldots \geq x_{n}\geq 0\\
\end{matrix}}\left\{\sum x_{i}^{\frac{2-\beta}{\beta}}\left(i+\sum_{j>i}\frac{x_{j}}{x_{i}}\right) \right\} \leq n^{\beta} a^{2-\beta}
\end{equation}
We will use induction on $n$ to prove our claim. The base case is trivial. We next perform the inductive step. 
\begin{align*}
f_{n+1}^{\beta}(a) & = \sup_{\begin{matrix}
\|x\|^{\frac{1}{\beta}}_{\frac{1}{\beta}}=a\\ 
x_{1}\geq \ldots \geq x_{n+1}\\
1\geq x_{i} \geq 0
\end{matrix}}\left\{x_{1}^{\frac{2-\beta}{\beta}} + x_{1}^{\frac{2(1-\beta)}{\beta}} \sum_{j=1}^{n} x_{j+1}+ \sum_{i=1}^{n} x_{i+1}^{\frac{2(1-\beta)}{\beta}}\left(i+1+\sum_{j>i}\left(\frac{x_{j+1}}{x_{i+1}}\right) \right) \right\}\\
& \leq  \sup_{\begin{matrix}
\|x\|^{\frac{1}{\beta}}_{\frac{1}{\beta}}=a\\ 
x_{1}\geq \ldots \geq x_{n+1}\\
1\geq x_{i} \geq 0
\end{matrix}}\left\{x_{1}^{\frac{2-\beta}{\beta}} + x_{1}^{\frac{2(1-\beta)}{\beta}} \|x_{-1}\|_{1}+ x_{1}^{\frac{1}{\beta}-1}\sum_{i=1}^{n} x_{i+1}^{\frac{1}{\beta}}+ f^{\beta}_{n}(a-x_{1}^{\frac{1}{\beta}}) \right\}\\
& \leq  \sup_{\tiny\begin{matrix} x_{1}\in \left[ \left(\frac{a}{n+1}\right)^{\beta}, \min\left\{\left(a - n\zeta \right)^{\beta},1\right\}\right]
\end{matrix}}\left\{ x_{1}^{\frac{2-\beta}{\beta}} + x_{1}^{\frac{2(1-\beta)}{\beta}} n^{1-\beta}(a - x_{1}^{\frac{1}{\beta}})^{\beta}+ x_{1}^{\frac{1}{\beta}-1}(a - x_{1}^{\frac{1}{\beta}})+  Cn^{\beta}(a-x_{1}^{\frac{1}{\beta}})^{2-\beta}\right\}
	\end{align*}
where in the last inequality we used H\"{o}lder's inequality (twice) and the inductive hypothesis. We perform the change of variables $x_{1}= (a\gamma)^{\beta}$ to obtain
\begin{align*}
f_{n+1}^{\beta}(a)&\leq (n+1)^{\beta}a^{2-\beta} \cdot  \sup_{\tiny\begin{matrix} \gamma\in\left[ \frac{1}{n+1}, 1\right] 
\end{matrix}}\left\{  \frac{1}{(n+1)^{\beta}}\gamma^{2-\beta}  +   \frac{n^{1-\beta}}{(n+1)^{\beta}} \gamma^{2-2\beta}  \left(1 - \gamma\right)^{\beta}\right.\\
&\qquad\qquad\qquad\qquad\qquad\qquad\qquad +\left. \frac{1}{(n+1)^{\beta}}\gamma^{1-\beta}\left(1 - \gamma\right) +  C\left(\frac{n}{n+1}\right)^{\beta}(1 - \gamma)^{2-\beta}\right\}
\end{align*}
To prove the desired inequality we  are going to show that the function
\[
g_{n}(\gamma) :=  \frac{1}{(n+1)^{\beta}}\gamma^{2-\beta}  +   \frac{n^{1-\beta}}{(n+1)^{\beta}} \gamma^{2-2\beta}  \left(1 - \gamma\right)^{\beta} + \frac{1}{(n+1)^{\beta}}\gamma^{1-\beta}\left(1 - \gamma\right) +  C\left(\frac{n}{n+1}\right)^{\beta}(1 - \gamma)^{2-\beta}
\] is monotone in the given interval. Its derivative is given by
\begin{align*}
g_{n}^{'}(\gamma) &=  \frac{1-\beta}{(n+1)^{\beta}}\gamma^{1-\beta} + 2(1-\beta)\frac{n^{1-\beta}}{(n+1)^{\beta}}\frac{(1-\gamma)^{\beta}}{\gamma^{2\beta -1}}+\frac{1-\beta}{(n+1)^{\beta}}\frac{1-\gamma}{\gamma^{\beta}}\\
& - \beta\frac{n^{1-\beta}}{(n+1)^{\beta}}\frac{\gamma^{2-2\beta}}{(1-\gamma)^{1-\beta}} - (2-\beta)C\frac{n^{\beta}}{(n+1)^{\beta}}(1-\gamma)^{1-\beta}
\end{align*}

Comparing negative and positive terms we get the following relationships:
\begin{align*}
  \beta\frac{n^{1-\beta}}{(n+1)^{\beta}}\frac{\gamma^{2-2\beta}}{(1-\gamma)^{1-\beta}}  \geq  \frac{1-\beta}{(n+1)^{\beta}}\gamma^{1-\beta}, \qquad\forall &\gamma\geq \frac{\left(\frac{1-\beta}{\beta}\right)^{\frac{1}{1-\beta}}}{n+\left(\frac{1-\beta}{\beta}\right)^{\frac{1}{1-\beta}}} \\
\frac{1}{2}\cdot (2-\beta)C\frac{n^{\beta}}{(n+1)^{\beta}}(1-\gamma)^{1-\beta}  \geq 2(1-\beta)\frac{n^{1-\beta}}{(n+1)^{\beta}}\frac{(1-\gamma)^{\beta}}{\gamma^{2\beta -1}}, \qquad &\forall\gamma\geq \frac{1}{1+\left(\frac{(2-\beta)C}{(1-\beta)4}\right)^{\frac{1}{2\beta-1}}n}\\
\frac{1}{2}\cdot (2-\beta)C\frac{n^{\beta}}{(n+1)^{\beta}}(1-\gamma)^{1-\beta} \geq \frac{1-\beta}{(n+1)^{\beta}}\frac{1-\gamma}{\gamma^{\beta}}, \qquad&\forall  \gamma\geq \frac{1}{1+\left(\frac{(2-\beta)C}{(1-\beta)2}\right)^{\frac{1}{\beta}}n}
\end{align*}
Our goal is to show that each lower bound on $\gamma$  is always less or equal to $\frac{1}{n+1}$ and hence that the function $g_{n}(\gamma)$ is monotone in the interval of interest. Starting, from the first equation we observe that the function $x/(1+x)$ is increasing in $x\geq 0$, so to maximize the lower bound we need to find the maximum of the function:
\[
h_{1}(\beta) = \left(\frac{1-\beta}{\beta}\right)^{\frac{1}{1-\beta}} = e^{\frac{\ln(1-\beta)-\ln(\beta)}{1-\beta}}
\]
The function is decreasing for any $\beta \in [\frac{1}{2},1]$ so we have that $\sup_{\beta\in[0.5,1]}h_{1}(\beta)=h_{1}(0.5) = 1$. This proves that the first relationship holds for all $\gamma \geq \frac{1}{n+1}$. To handle the other two inequalities, we observe that the function $1/(1+x)$ is decreasing for $x\geq 0$ and hence we need to find the minimum of the functions:
\begin{align}
h_{2}(\beta)  &= \left(r_{2}\frac{(2-\beta)}{1-\beta}\right)^{\frac{1}{2\beta-1}} = e^{\frac{\ln(2-\beta)-\ln(1-\beta)+\ln r_{2}}{2\beta-1}}\\
h_{3}(\beta) & = \left(r_{3}\frac{(2-\beta)}{1-\beta}\right)^{\frac{1}{\beta}} = e^{\frac{\ln(2-\beta)-\ln(1-\beta)+\ln r_{3}}{\beta}}
\end{align}
where $r_{2},r_{3}>0$ are constants. The function $h_{2}(\beta)$ is increasing in the interval $[0.5,1]$, hence we get the desired inequality $ h_{2}(\beta)\geq h_{2}(0.5) = 1$ . Similarly, the function $h_{3}(\beta)$ satisfies $h_{3}(\beta)\geq h_{3}(0.5) = (3r_{3})^{2}=(3C/2)^{2}$, which is greater than one for $C=1$. This shows that the function $g_{n}(\gamma)$ is decreasing and consequently that
\[
f_{n+1}^{\beta}(a) \leq a^{2-\beta}(n+1)^{\beta}\cdot g_{n}\left(\frac{1}{n+1}\right)= a^{2-\beta}(n+1)^{\beta}
\]
completing the inductive step.
\end{proof} 
\section*{Acknowledgements}
We thank Peter Bailis for valuable conversations that initiated our research on this problem. This research was supported by NSF grants CCF-1617577, CCF-1302518 and a Simons Investigator
Award. Paris Siminelakis is partially supported by an Onassis Foundation Scholarship.

\bibliography{kde}
\bibliographystyle{abbrv}
\appendix
\section{Chernoff Bounds}
For completeness we state here the version of Chernoff bounds that we use (e.g.~\cite{mitzenmacher2005probability}).
\begin{theorem}
Let $X_{1},\ldots, X_{n}$ be i.i.d Bernouli random variables with mean $\E[X_{i}]=\bar{x}$, then for $\delta\in (0,1]$
\begin{align}
\P[\sum_{i=1}^{n}X_{i} \leq (1-\delta)n\bar{x}] &\leq e^{-\frac{\delta^{2}n\bar{x}}{2}}\\
\P[\sum_{i=1}^{n}X_{i} \geq (1+\delta)n\bar{x}] &\leq e^{-\frac{\delta^{2}n\bar{x}}{3}}
\end{align}
\end{theorem} 
\section{Random Fourier Features Estimator}\label{sec:RFF}
The basis of using random Fourier features is Bochner's theorem~\cite{bochner1933monotone}, that for the gaussian kernel us equivalent to.
\begin{fact} Let $\omega\sim \mathcal{N}(0,I_{d})$, 
$\E[\cos(\omega^{\top}(x-y))]=e^{-\|x-y\|^{2}}$
\end{fact}
The \emph{random fourier features} basic estimator is given by
\[
Z_{RFF}(x) := \frac{2}{|P|}\sum_{y \in P} \cos(\omega^{\top}x +b) \cos(\omega^{\top}y+b)
\] where $\omega\sim \mathcal{N}(0,I_{d})$ and $b\sim U[0,2\pi]$
\begin{proposition}\label{prop:RFF-var}
The first two moments of the RFF estimator are given by
\begin{eqnarray}
\E[Z_{RFF}(x)] &=& \frac{1}{|P|}\sum_{y\in P}k(x,y)\\
\E[Z_{RFF}^{2}]&=&  \frac{1}{|P|} + \frac{1}{|P|}\sum_{y\in P} k(2x,2y) + \frac{1}{|P|^{2}}\sum_{y\neq z\in P}k(y,z) + \frac{1}{|P|^{2}}\sum_{y\neq z\in P} k(2x,y+z)
\end{eqnarray}
\end{proposition}

To prove the lemma we require the following trigonometric identities.
\begin{proposition}[Random angle]\label{prop:angles}
 Let $B\sim U[0,2\pi]$ and $a,b,c\in \R$, we have the following identities:
\begin{eqnarray}
\E_{B}[\cos(a+B)\cos(c+B)] &=& \frac{1}{2}\cos(a-c)\label{eq:A}\\
\E_{B}[\cos^{2}(a+B)\cos^{2}(c+B)] &=& \frac{1}{4}+\frac{1}{8}\cos(2a-2c)\label{eq:B}\\
\E_{B}[\cos^{2}(a+B)\cos(b+B)\cos(c+B)] &=&\frac{1}{4}\cos(b-c)+\frac{1}{8}\cos(2a-(b+c))\label{eq:C}
\end{eqnarray}
\end{proposition}

\begin{proof}[Proof of Proposition \ref{prop:RFF-var}]
We have that $\E[Z_{RFF}(x)] = \textrm{KDE}_{P}(x) =  \frac{1}{|P|}\sum_{y\in P} k(x,y)$  Further
\[
\E[Z^{2}_{RFF}] = \frac{1}{|P|^{2}}\E\left\{\left|\sum_{y\in P}Z(x,y)\right|^{2}\right\} 
= \frac{1}{|P|^{2}}\left\{\sum_{y\in P} \E[Z(x,y)^{2}] + \sum_{y\neq z\in P} \E[Z(x,y)Z(x,z)]\right\}\]
where $Z(x,y) := 2 \cos(\omega^{\top}x+b)\cos(\omega^{\top}y+b)$.  We start with the first term:
\begin{align*}
\E[Z(x,y)^{2}] & = 4\E[ \cos^{2}(\omega^{\top}x+b)\cos^{2}(\omega^{\top}y+b)] = 1 + \frac{1}{2}k(2x,2y)
\end{align*}
Next, we analyze the second term 
\begin{equation}
\E[Z(x,y)Z(x,z)] = 4\E[\cos^{2}(\omega^{\top}x+b)\cos(\omega^{\top}y+b)\cos(\omega^{\top}z+b)] = k(y,z) + \frac{1}{2}k(2x,y+z)
\end{equation}
The second moment thus becomes:
\begin{align}
\E[Z^{2}_{RFF}] = \frac{1}{|P|} + \frac{1}{|P|}\sum_{y\in P} k(2x,2y) + \frac{1}{|P|^{2}}\sum_{y\neq z\in P}k(y,z) + \frac{1}{|P|^{2}}\sum_{y\neq z\in P} k(2x,y+z)
\end{align}
Imagine, that $P$ consists of $n$ points all placed at distance $\sqrt{\log(1/\mu)}$ of $x$. In this case the term $\frac{1}{|P|}\sum_{y\neq z\in P}k(y,z)$ is constant independent of $\mu$. So, when the dataset is highly correlated the variance of the estimator is very large.
\end{proof}
\begin{proof}[Proof of Proposition \ref{prop:angles}] Using partial integration and the fact that $\sin,\cos$ are periodic functions with period $2\pi$, we get
\begin{eqnarray*}
\E_{B}[\sin(a+B)\sin(c+B)]&=& - \left[\cos(a+B)\sin(c+B)\right]_{0}^{2\pi}+\E_{B}[\cos(a+B)\cos(c+B)]\\
&=& \E_{B}[\cos(a+B)\cos(c+B)]
\end{eqnarray*}
The first equation then follows from the identity $\cos(A)\cos(B)=\cos(A-B)-\sin(A)\sin(B)$. To get the second equation we use the identities  $\cos^{2}(a)=\frac{1}{2}(1+\cos(2a))$  and  $\E_{B}[\cos(2a+2B)]=0$ to get
\[
\E_{B}[\cos^{2}(a+B)\cos^{2}(c+B)]=\frac{1}{4}+\frac{1}{4}\E_{B}[\cos(2a+2B) \cos(2c+2B)]
\]\
Finally, invoking the first part completes the proof.
\end{proof}	
\paragraph{Cost of Linearity} Looking at the RFF estimator one might think that it requires linear time per query. However, observe that if we spend linear time once to compute  $\phi_{P} = \frac{1}{|P|}\sum_{y\in P} \sqrt{2} \cos(\omega^{\top}y+\theta)$, the estimate $Z_{RFF}(x) = \sqrt{2} \cos(\omega^{\top}x) \cdot \phi_{P}$ can be computed in $O(1)$ time per query. Perhaps counter-intuitively this shows that a simple rescaling gives an unbiased estimate for all queries $x$! Thus, it is not surprising that the variance of the estimator is essentially as large as possible. If we want to reduce the variance, we can take the average of $m$ such estimates, that would be equivalent to constructing a randomized embedding $\phi:\R^{d}\to \R^{m}$  of the points into a euclidean space such that $\mathrm{KDE}_{P}(x) \approx \langle\phi(x), \frac{1}{|P|}
\sum_{y\in P}\phi(y)\rangle$. This approach is quite general and can be extended to a large range of kernels \cite{rahimi2007random,kar2012random,pham2013fast}. However, the inherent variance of the randomized embedding suggests that this approach cannot beat even the  Random Sampling estimator.

\end{document}